\documentclass[10pt, final, journal, letterpaper, twocolumn]{IEEEtran}

\usepackage{bbm}
\usepackage{mathrsfs}
\usepackage{amsfonts}
\usepackage[dvips]{graphicx}
\usepackage{times}
\usepackage{cite}
\usepackage{amsmath}
\usepackage{array}
\usepackage{amssymb}

\usepackage{stfloats}
\usepackage{slashbox}
\usepackage{graphicx}
\usepackage{footnote}
\usepackage{booktabs}
\usepackage{algorithmic}
\usepackage{algorithm}
\usepackage{color}
\usepackage{cases}

\newtheorem{proposition}{Proposition}

\begin{document}

\title{Cross-Layer Optimization of Two-Way Relaying for Statistical QoS Guarantees}

\author{Cen~Lin, Yuan~Liu,~\IEEEmembership{Student~Member,~IEEE}, and Meixia~Tao,~\IEEEmembership{Senior~Member,~IEEE}
\thanks{Manuscript received August 29, 2011; revised March 4, 2012.}
\thanks{The authors are with the Department of
Electronic Engineering at Shanghai Jiao Tong University, Shanghai,
200240, P. R. China. Email: \{lincen,yuanliu, mxtao\}@sjtu.edu.cn.}
\thanks{This work is supported by the National 973 project under grant 2012CB316100, the NSFC
under grant 60902019, and the NCET Program under grant NCET-11-0331. This paper was
presented in part at the IEEE ICC, Ottawa, Canada, June 2012 \cite{Lin}.}}

 \maketitle

\begin{abstract}

Two-way relaying promises considerable improvements on spectral
efficiency in wireless relay networks. While most existing works
focus on physical layer approaches to exploit its capacity gain, the
benefits of two-way relaying on upper layers are much less
investigated. In this paper, we study the cross-layer design and
optimization for delay quality-of-service (QoS) provisioning in
two-way relay systems. Our goal is to find the optimal transmission
policy to maximize the weighted sum throughput of the two users in
the physical layer while guaranteeing the individual statistical
delay-QoS requirement for each user in the datalink layer. This
statistical delay-QoS requirement is characterized by the QoS
exponent. By integrating the concept of effective capacity, the
cross-layer optimization problem is equivalent to a weighted sum
effective capacity maximization problem. We derive the jointly
optimal power and rate adaptation policies for both three-phase and
two-phase two-way relay protocols. Numerical results show that the
proposed adaptive transmission policies can efficiently provide QoS
guarantees and improve the performance. In addition, the throughput
gain obtained by the considered three-phase and two-phase protocols
over direct transmission is significant when the delay-QoS
requirements are loose, but the gain diminishes at tight delay
requirements. It is also found that, in the two-phase protocol, the
relay node should be placed closer to the source with more stringent
delay requirement.

\end{abstract}


\section{Introduction}
\setlength\arraycolsep{2pt}

The explosive developments of wireless communication have brought us
into a new era where higher data transmission rates and diverse
quality-of-service (QoS) provisioning are desperately expected.
Real-time applications, such as voice over IP and video streaming,
which are highly delay-sensitive, need reliable QoS guarantees. The
design merely at the physical layer may not ensure the desired QoS
requested by the service from upper layers. Only through the
interaction and optimization between different layers can such QoS
guarantees be fulfilled. This kind of cross-layer approach relaxes
the layering architecture of the conventional network model and
brings remarkable performance enhancement, which in turn could
result in high complexity. Therefore, to develop efficient
cross-layer methods with small information flows between layers is
interesting from both theoretical and practical perspectives.

Recently, two-way relaying has appeared as an advanced relay
technique to significantly boost the spectral efficiency in wireless
networks \cite{Zhang, Popovski, Rankov, Kim}. The notion of two-way
relaying is to apply the principle of network coding at the physical
layer so that only three or two time slots are needed when a pair of
nodes exchange information via a relay node, while the conventional
one-way relaying requires four time slots. Most previous efforts on
two-way relaying have focused on the design and optimization merely
in physical layer, such as analysis of capacity bounds
\cite{Kim,Liu}, adaptive network-coded constellation mapping
\cite{Akino}, joint channel coding and network coding design
\cite{Sheng,Tao}, precoding design with multiple antennas
\cite{Rui}, and resource allocation for throughput maximization in
OFDMA networks \cite{Jitvanichphaibool,Chen,Yuan,Mei}. Certainly, it
is desirable and promising to investigate the cross-layer design and
optimization of two-way relay architecture for QoS provisioning. To
our best knowledge, only two attempts have been made to study the
cross-layer design for two-way relaying \cite{Oechtering,
Ciftcioglu}. Specifically, \cite{Oechtering} and \cite{Ciftcioglu}
characterized the queue stability region for infinite backlogs with
the XOR and superposition-coding based decode-and-forward (DF)
protocols, respectively. Nevertheless, having stable queue does not
provide optimality in delay performance.

Motivated by the needs for bounded delay performance and small
inter-layer information flows, we adopt the concept of \emph{QoS
exponent} and consider the cross-layer optimization of two-way relay
systems for statistical delay QoS guarantees in this work. The QoS
exponent is used to characterize the statistical delay performance
metric, namely, the delay-bound violation probability, and is the
only requested information exchanged between the datalink layer and
the physical layer \cite{Chang}. As a result, our work builds on the
information theory to capture the performance limits at the physical
layer and the statistical QoS theory to model the delay performance
from upper layers. Specifically, we aim to find the optimal
transmission strategies to maximize the weighted sum rate while
guaranteeing the individual statistical delay QoS requirement for
each source node. Through integrating the theory of \emph{effective
capacity} \cite{Wu}, we convert this problem into a weighted sum
effective capacity maximization problem. After that, the optimal
power and rate adaptation policies as functions of both the network
channel state information (CSI) and the delay-QoS constraints are
developed.

The main contributions and results of this paper are summarized as
follows: We formulate the cross-layer optimization problem for
statistical QoS guarantees in two-way relay systems as weighted sum
effective capacity maximization. This problem is shown to be convex.
Furthermore, we propose the optimal transmission strategies with
joint power and rate adaptation for both three-phase and two-phase
two-way relaying schemes. Particularly, for the two-phase protocol,
the optimal channel state partition criterion for successive
decoding in the multiple-access (MAC) phase is derived. Numerical
results show that, compared with the conventional two-way direct
transmission, the considered three-phase and two-phase two-way relay
protocols can significantly improve the average throughput when the
statistical delay QoS requirements are loose. However, as delay
constraints become stringent, the performance gain reduces, and
eventually all the throughputs approach zero. It is also
demonstrated that, under the same delay-QoS constraint, the
three-phase protocol has higher weighted sum effective capacity than
the two-phase protocol in high signal-to-noise ratio (SNR) regime,
but has lower weighted sum effective capacity than the two-phase
protocol in low SNR regime. Moreover, we show that the three-phase
protocol has superiority over the two-phase protocol when the relay
is extremely adjacent to either of the sources. Meanwhile, it is
better to place the relay closer to the source with more stringent
delay requirement for the two-phase protocol.

The remainder of this paper is organized as follows. Section II
introduces some preliminaries on statistical QoS guarantees and the
related work. Section III presents the system model and demonstrates
our problem formulation. In Section IV and Section V, we propose the
optimal transmission strategies for the three-phase and two-phase
transmission, respectively. Numerical results are provided in
Section VI to verify the effectiveness of the proposed policies.
Finally, we conclude the paper in Section VII.

\section{Background on Statistical QoS Guarantees and Related Work}

\subsection{Preliminaries on Statistical QoS Guarantees}

Due to the time-varying nature of wireless channels, it is
infeasible to guarantee the hard delay bound for real-time traffic.
Therefore, statistical QoS metric, in the form of the delay-bound
violation probability, is commonly used to characterize the diverse
delay-QoS requirements.

Based on the large deviation principle, the author in \cite{Chang}
showed that under sufficient conditions, the stationary queue length
process $Q(t)$ converges in distribution to a random variable
$Q(\infty)$ satisfying that
\begin{equation}
-\lim_{Q_{th}\to\infty}\frac{\ln(Pr\{Q(\infty)>Q_{th}\})}{Q_{th}}=\theta,
\end{equation}
where $\theta>0$ is called \emph{QoS exponent}, denoting the
exponential decay rate of the distribution, and $Q_{th}$ is the
queueing length bound. According to the above equation, the
probability that the steady-state queue length exceeds a certain
bound $Q_{th}$ can be approximated by
\begin{equation}
Pr\{Q(\infty)>Q_{th}\}\approx e^{-\theta Q_{th}}.
\end{equation}
Similarly, the delay-bound violation probability can be stated as
\begin{equation}\label{e42}
Pr\{D>D_{th}\}\approx e^{-\theta \varphi(\theta)D_{th}},
\end{equation}
where $D$ and $D_{th}$ denote the queueing delay and delay bound,
respectively, and $\varphi(\theta)$ is known as the \emph{effective
bandwidth} of the arrival process under a given $\theta$. From
(\ref{e42}) we can conclude that the violation probability for a
given delay bound is characterized by the QoS exponent $\theta$.
Therefore, the dynamics of $\theta$ correspond to different delay
requirements. Obviously, a smaller $\theta$ implies a looser delay
QoS constraint, while a larger $\theta$ means a more strict delay
QoS constraint. In particular, when $\theta \rightarrow 0$, the
queueing system can tolerate an arbitrary delay, whereas when
$\theta \rightarrow \infty$, the system cannot allow any delay.

Inspired by the theory of effective bandwidth, Wu and Negi
introduced \emph{effective capacity} in \cite{Wu}, which is defined
as the maximum constant arrival rate that a given service process
can support in order to guarantee a statistical QoS requirement
specified by $\theta$. Analytically, the effective capacity, denoted
by $\mathscr{C}_{\emph{e}}(\theta)$, can be given by
\begin{equation}
\mathscr{C}_{\emph{e}}(\theta)=-\lim_{t\to\infty}\frac{1}{\theta
t}\ln\big(\mathbb{E}\big[e^{-\theta S[t]}\big]\big),
\end{equation}
where $S[t]=\sum_{i=1}^{t}R[i]$ is the time-accumulation of the
service process, and $\{R[i],i=1,2,\ldots\}$ corresponds to the
discrete-time stationary and ergodic stochastic service process.
$\mathbb{E}[\cdot]$ denotes the expectation.

Under the assumption that the block fading channel is independent
and identically distributed (i.i.d) over each time frame, the
sequence $\{R[i]\}$ is uncorrelated. Then the effective capacity can
be rewritten as
\begin{equation}
\mathscr{C}_{\emph{e}}(\theta)=-\frac{1}{\theta}\ln\big(\mathbb{E}\big[e^{-\theta
R[i]}\big]\big).
\end{equation}
Since the average arrival rate is equal to the average service rate
when the queue is in steady state, effective capacity can also be
regarded as the maximum throughput subject to a delay-QoS
constraint. In particular, as $\theta \rightarrow 0$, the optimal
effective capacity approaches the ergodic capacity of the channel.
On the other hand, as $\theta \rightarrow \infty$, the optimal
effective capacity is drawing to the zero-outage capacity of the
channel.

\subsection{Related Work on Delay-QoS Provisioning}

Delay-constrained cross-layer optimization has been studied
extensively in a variety of wireless networks
\cite{Berry,Wang,Hui,Liang,Jia,Tang,Du,Miller,Qiao}. For instance,
\cite{Berry} focused on the characterization of the stability region
and throughput optimal control for one-way relay systems. The delay
minimization problem based on power control and relay selection for
one-way relay systems with multiple antennas was considered in
\cite{Wang}. Joint power and subcarrier allocation for conventional
OFDMA networks with heterogeneous delay constraints was explored in
\cite{Hui,Liang}. Power allocation with statistical delay-QoS
provisioning for conventional point-to-point, one-way relaying, and
multiuser systems was studied in \cite{Jia}, \cite{Tang} and
\cite{Du}, respectively. Authors in \cite{Miller} proposed an
optimal scheduling algorithm for time division based multiuser
systems for statistical delay guarantees. Successive decoding order
optimization with fixed power assignment for MAC channel under
statistical delay constraints was investigated in \cite{Qiao}.

In view of all these existing literature, only the problems for
unidirectional communication were addressed, while the bidirectional
nature of the networks has not been fully exploited for delay-QoS
provisioning. Moreover, the impact of the cross-layer design and
optimization under delay constraints in two-way relaying has not
been revealed. Therefore, it is of great importance and necessity to
investigate the two-way relay networks for delay-QoS provisioning.

\section{System Overview and Problem Formulation}

\subsection{System Model}

The cross-layer two-way relay system is shown in
Fig.~\ref{fig:model}, where two source nodes, $A$ and $B$, exchange
messages via the relay node $R$. Each node operates in a half-duplex
manner. Like in \cite{Oechtering, Tang}, we consider that the
packets arriving at the relay node are forwarded immediately. As
illustrated in Fig.~\ref{fig:model}, in the datalink layer, two
first-in first-out (FIFO) queues are implemented at the two sources,
which consist of packets from upper layer to be transmitted. For QoS
provisioning, the packets transmitted from one source node to the
other are subject to the delay constraints, i.e., $\theta_{A}$ and
$\theta_{B}$. In the physical layer, each data packet is divided
into frames. Each frame is further divided into three or two slots
depending on the two-way relay protocols to be discussed in Section
III-C.

\begin{figure}[!t]
\begin{centering}
\includegraphics[scale=.55]{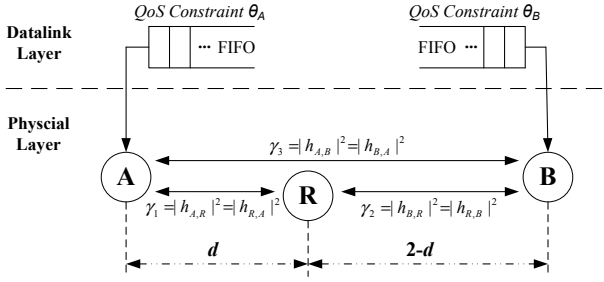}
\vspace{-0.1cm}
 \caption{Cross-layer two-way relay model.}\label{fig:model}
\end{centering}
\vspace{-0.3cm}
\end{figure}

\subsection{Channel Model}

We consider the scenario in which all nodes have perfect channel
state information. The channel coefficients of all links are assumed
to remain unchanged within each time frame but vary from one frame
to another. The instantaneous channel coefficient between node $i$
and $j$ is denoted as $h_{i,j}$, where $i,j \in \{A,B,R\}$ with
$i\neq j$. Here the channel reciprocity is assumed, i.e.,
$h_{i,j}=h_{j,i}$, which is valid in time-division duplex mode.
Without loss of generality, it is assumed that the additive noises
at all nodes are independent circularly symmetric complex Gaussian
random variables, each having zero mean and unit variance. For
notational convenience, we further define the instantaneous network
channel state information as a three-tuple
$\boldsymbol{\gamma}=(\gamma_1,\gamma_2, \gamma_3)$, where
$\gamma_{1}=|h_{A,R}|^2, \gamma_{2}=|h_{B,R}|^2,
\gamma_{3}=|h_{A,B}|^2$ (as shown in Fig.~\ref{fig:model}).

\subsection{Two-Way Relay Protocols}

Different two-way relay protocols have been studied in the
literature \cite{Popovski,Rankov,Kim,Liu,Tao}. In this paper, we
focus on the three-phase and two-phase two-way relay protocols with
DF strategy. Let $P_A$, $P_B$, and $P_R$ denote the transmit power
of nodes $A$, $B$, and $R$, respectively. Let $R_{A}$ and $R_{B}$
denote the achievable rates from node $A$ to node $B$ and from node
$B$ to node $A$, respectively. We further denote
$\boldsymbol{P}=[P_{A},P_{B},P_{R}]^T$ as the transmit power vector,
and $\boldsymbol{R}=[R_{A},R_{B}]^T$ as the bidirectional rate pair.

\begin{figure}[!t]
\begin{centering}
\includegraphics[scale=.59]{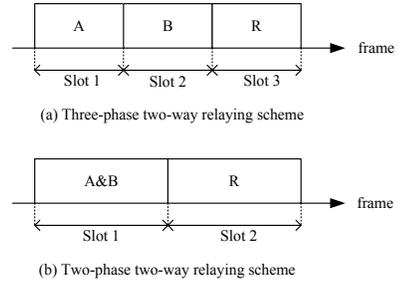}
\vspace{-0.1cm}
 \caption{Two-way relay transmission schemes.}\label{fig:scheme}
\end{centering}
\vspace{-0.3cm}
\end{figure}

\subsubsection{Three-Phase Two-Way Relaying}

In this protocol, the information exchange between A and B is
completed in three time slots. As shown in Fig.~\ref{fig:scheme}(a),
in the first time slot, source node $A$ transmits, while source node
$B$ and relay node $R$ listen. In the second time slot, node $B$
transmits, while $A$ and $R$ listen. In the third time slot, the
relay node transmits and both $A$ and $B$ listen. Ideally, the time
fraction of each slot in every transmission frame can be optimized.
In this work we only consider equal time assignment for simplicity,
and so is for the two-phase protocol.

The achievable rate region of the three-phase protocol with DF
strategy under a given transmit power vector $\boldsymbol{P}$ and
network CSI $\boldsymbol{\gamma}$ is \cite{Kim}
\begin{eqnarray}\label{eqn:e1}
\mathcal{R}(\boldsymbol{P},\boldsymbol{\gamma})=\{(R_{A},R_{B})\},
\end{eqnarray}
where
\begin{small}
\begin{numcases}
{R_{A}\leq}\frac{1}{3}\min\big\{C(\gamma_{1}P_{A}),C(\gamma_{3}P_{A})+C(\gamma_{2}P_{R})\big\},
&\hspace{-0.25cm}$\gamma_{1}>\gamma_{3}${\normalsize~\hss{(7)}}\nonumber\\
\frac{1}{3}C(\gamma_{3}P_{A}), & \hspace{-0.51cm}${\rm
otherwise}${\normalsize~\hss{(8)}}\nonumber
\end{numcases}
\end{small}

\begin{small}
\begin{numcases}{R_{B}\leq}
\frac{1}{3}\min\big\{C(\gamma_{2}P_{B}),C(\gamma_{3}P_{B})+C(\gamma_{1}P_{R})\big\},
& \hspace{-0.3cm}$\gamma_{2}>\gamma_{3}${\normalsize~\hss{(9)}}\nonumber\\
\frac{1}{3}C(\gamma_{3}P_{B}), & \hspace{-0.71cm}${\rm otherwise}${\normalsize~\hss{(10)}}\nonumber
\end{numcases}
\end{small}

\hspace{-0.35cm}with $C(x)=\log_{2}(1+x)$. It shows that if the
channel quality of the relay link for one data flow ($\gamma_1$ or
$\gamma_2$) is worse than that of the direct link ($\gamma_3$), then
direct transmission will be triggered for that flow.

\setcounter{equation}{10}

\subsubsection{Two-Phase Two-Way Relaying}

In this scheme, it takes two slots to finish one round of
information exchange between the two source nodes. As shown in
Fig.~\ref{fig:scheme}(b), in the first time slot, which is termed as
multiple access (MAC) phase, the nodes $A$ and $B$ simultaneously
transmit signals to the relay node $R$. Due to the half-duplex
constraint, there is no direct link between nodes $A$ and $B$. In
the second time slot, which is known as broadcast (BC) phase, the
relay node broadcasts signals to both $A$ and $B$.

The achievable rate region of the two-phase two-way relaying with DF
strategy under a given transmit power vector $\boldsymbol{P}$ and
network CSI $\boldsymbol{\gamma}$ is \cite{Popovski,Rankov,Kim}

\begin{equation}\label{eqn:e6}
\mathcal{R}(\boldsymbol{P},\boldsymbol{\gamma})=\mathcal{C}_{MAC}(P_{A},P_{B},\boldsymbol{\gamma})\cap\mathcal{C}_{BC}(P_{R},\boldsymbol{\gamma}),
\end{equation}
where $\mathcal{C}_{MAC}$ and $\mathcal{C}_{BC}$ are the achievable
rate regions of the MAC and BC phases, respectively, and can be
given by
\begin{eqnarray}
\mathcal{C}_{MAC}(P_{A},P_{B},\boldsymbol{\gamma})=\{(R_{A},R_{B}):R_{A}\leq
\frac{1}{2}C(\gamma_{1}P_{A}), \nonumber\\ R_{B}\leq
\frac{1}{2}C(\gamma_{2}P_{B}), R_{A}+R_{B}\leq
\frac{1}{2}C(\gamma_{1}P_{A}+\gamma_{2}P_{B}) \},
\end{eqnarray}
\begin{eqnarray}\label{eqn:e7}
\mathcal{C}_{BC}(P_{R},\boldsymbol{\gamma})=\{(R_{A},R_{B}):&&R_{A}\leq
\frac{1}{2}C(\gamma_{2}P_{R}),\nonumber\\ &&R_{B}\leq
\frac{1}{2}C(\gamma_{1}P_{R})\}.
\end{eqnarray}
Note that both $\mathcal{C}_{MAC}$ and $\mathcal{C}_{BC}$ are convex
sets, so is their intersection.

\subsection{Problem Formulation}

In this paper, our objective is to find the optimal transmission
policies to maximize the weighted sum rate of the two-way relay
system while satisfying the individual delay requirement at each
node. As stated before, the effective capacity can be viewed as the
maximum throughput under the constraint of QoS exponent in steady
state. Hence, we can formulate an equivalent problem, which is to
maximize the weighted sum effective capacity for given delay
constraints of node $A$ and node $B$, i.e., $\theta_{A}$ and
$\theta_{B}$. Our resource allocation policies are based on
cross-layer parameters, specifically, the instantaneous network CSI
$\boldsymbol{\gamma}=(\gamma_1,\gamma_2, \gamma_3)$ and the
delay-QoS requirements
$\boldsymbol{\theta}=(\theta_{A},\theta_{B})$. Correspondingly, we
define
$\boldsymbol{\epsilon}\triangleq(\boldsymbol{\gamma},\boldsymbol{\theta})$.
Therefore, the problem can be formulated as follows,
\begin{eqnarray}
\textbf{P1}:&&\max_{\boldsymbol{P}(\boldsymbol{\epsilon}),\boldsymbol{R}(\boldsymbol{\epsilon})}~-\frac{\omega_{A}}{\theta_{A}}
\ln\big(\mathbb{E}_{\boldsymbol\gamma}[e^{-\theta_{A}R_{A}(\boldsymbol{\epsilon})}]\big)\nonumber\\
&&~~~~~~~~~~~~-\frac{\omega_{B}}{\theta_{B}}\ln\big(\mathbb{E}_{\gamma}[e^{-\theta_{B}R_{B}(\boldsymbol{\epsilon})}]\big)\label{eqn:e11}\\
&&~~~s.t.~~~~~~\mathbb{E}_{\boldsymbol\gamma}[P_{A}]\leq \overline{P_{A}}\label{eqn:e8}\\
&&~~~~~~~~~~~~~\mathbb{E}_{\boldsymbol\gamma}[P_{B}]\leq \overline{P_{B}}\label{eqn:e9}\\
&&~~~~~~~~~~~~~\mathbb{E}_{\boldsymbol\gamma}[P_{R}]\leq \overline{P_{R}}\label{eqn:e10}\\
&&~~~~~~~~~~~~~\boldsymbol{R}(\boldsymbol{\epsilon})\in
\mathcal{R}(\boldsymbol{P}(\boldsymbol{\epsilon}),\boldsymbol{\gamma})\label{eqn:e12}\\
&&~~~~~~~~~~~~~\boldsymbol{P}(\boldsymbol{\epsilon})\succeq
0,\label{eqn:e13}
\end{eqnarray}
where $\omega_{A},\omega_{B}$ are the weights assigned to the two
users, satisfying $\omega_{A}+\omega_{B}=1$, and
$\overline{P_{A}},\overline{P_{B}},\overline{P_{R}}$ are the
long-term power constraints of node $A$, $B$ and $R$, respectively.
$\mathbb{E}_{\boldsymbol\gamma}[\cdot]$ emphasizes that the
expectation is with regard to $\boldsymbol{\gamma}$.
$\boldsymbol{P}(\boldsymbol{\epsilon})$ and
$\boldsymbol{R}(\boldsymbol{\epsilon})$ denote the power and rate
adaptation policies to be optimized, which are functions of
$\boldsymbol{\epsilon}$. The instantaneous rate region
$\mathcal{R}(\boldsymbol{P}(\boldsymbol{\epsilon}),\boldsymbol{\gamma})$
is defined in (\ref{eqn:e1}) for the three-phase protocol, or in
(\ref{eqn:e6}) for the two-phase protocol. Note that
$\mathcal{R}(\boldsymbol{P}(\boldsymbol{\epsilon}),\boldsymbol{\gamma})$
is a convex space spanned by the power sets
$\boldsymbol{P}(\boldsymbol{\epsilon})$.

It is proved in \cite{Du} that the weighted sum effective capacity
is a concave function of the powers in multiuser systems with direct
transmission. Using the similar method, we can prove the weighted
sum effective capacity in the two-way relay system as given in
(\ref{eqn:e11}) is also concave of
$\boldsymbol{P}(\boldsymbol{\epsilon})$. The main reason is that the
achievable rate pair $R_{A}$ and $R_{B}$ are both concave with
respect to the power vector $\boldsymbol{P}$. In addition, the power
constraints (\ref{eqn:e8})-(\ref{eqn:e10}) and (\ref{eqn:e13}) are
affine. Thus, the problem $\textbf{P1}$ is a convex optimization
problem, and there exists a unique and optimal solution. In the next
two sections, we shall develop the optimal power and rate adaptation
policies of the given problem for the three-phase and two-phase
protocols, respectively.

\begin{figure}[!t]
\begin{centering}
\includegraphics[scale=.6]{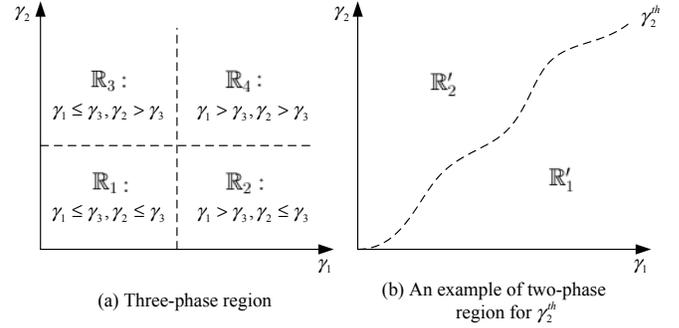}
\vspace{-0.01cm}
 \caption{Cross-layer two-way relaying channel state region.}\label{fig:csregion}
\end{centering}
\vspace{-0.3cm}
\end{figure}

\section{Optimal Policy for Three-Phase Two-way Relaying}

In this section, we first derive the optimal transmission policy for
the three-phase two-way relaying subject to general delay QoS
requirements $\boldsymbol{\theta}=(\theta_{A},\theta_{B})$, for
which
$\mathcal{R}(\boldsymbol{P}(\boldsymbol{\epsilon}),\boldsymbol{\gamma})$
is given in (\ref{eqn:e1}). Then we consider the limiting case when
both $\theta_{A}$ and $\theta_{B}$ approach zero, i.e., the ergodic
capacity problem.

\subsection{Optimal Policy}

We define the Lagrangian of problem \textbf{P1} as
\begin{eqnarray}\label{eqn:e30}
&&\mathcal
{L}\big(\boldsymbol{P}(\boldsymbol{\epsilon}),\boldsymbol{R}(\boldsymbol{\epsilon}),\boldsymbol{\lambda}\big)\nonumber\\
&=&-\frac{\omega_{A}}{\theta_{A}}
\ln\big(\mathbb{E}_{\boldsymbol\gamma}[e^{-\theta_{A}R_{A}(\boldsymbol{\epsilon})}]\big)-\frac{\omega_{B}}{\theta_{B}}\ln\big(\mathbb{E}_{\gamma}[e^{-\theta_{B}R_{B}(\boldsymbol{\epsilon})}]\big)\nonumber\\
&&+\lambda_{A}\big(\overline{P_{A}}-\mathbb{E}_{\boldsymbol\gamma}[P_{A}(\boldsymbol{\epsilon})]\big)+\lambda_{B}\big(\overline{P_{B}}-\mathbb{E}_{\boldsymbol\gamma}[P_{B}(\boldsymbol{\epsilon})]\big)\nonumber\\
&&+\lambda_{R}\big(\overline{P_{R}}-\mathbb{E}_{\boldsymbol\gamma}[P_{R}(\boldsymbol{\epsilon})]\big),
\end{eqnarray}
where $\boldsymbol{\lambda}=[\lambda_A,\lambda_B,\lambda_R]^T$ are
the Lagrange multipliers related to the power constraints
(\ref{eqn:e8})-(\ref{eqn:e10}). Then the dual problem of \textbf{P1}
can be stated as
\begin{eqnarray}
\textbf{P2}:~&&\min_{\boldsymbol{\lambda}\geq0}~\max_{\begin{subarray}{c}\boldsymbol{P}(\boldsymbol{\epsilon})\succeq0,\\\boldsymbol{R}(\boldsymbol{\epsilon})\in
\mathcal{R}\end{subarray}}~\mathcal
{L}\big(\boldsymbol{P}(\boldsymbol{\epsilon}),\boldsymbol{R}(\boldsymbol{\epsilon}),\boldsymbol{\lambda}\big).\nonumber
\end{eqnarray}
Note that the subgradient method can be used to update
$\boldsymbol{\lambda}$ toward the optimal $\boldsymbol{\lambda}^*$
as follows
\begin{eqnarray}
\boldsymbol{\lambda}^{(i+1)}=\boldsymbol{\lambda}^{(i)}-\boldsymbol{s}^{(i)}\big(\boldsymbol{\overline{P}}-\mathbb{E}_{\boldsymbol\gamma}[\boldsymbol{P}(\boldsymbol{\epsilon})]\big),
\end{eqnarray}
where the subscript $i$ denotes the iteration index, and
$\boldsymbol{s}^{(i)}$ is the vector of step size designed properly.

According to the achievable rate region defined in (\ref{eqn:e1}),
we divide the channel states $\boldsymbol\gamma$ into four regions
as shown in Fig.~\ref{fig:csregion}(a). In the following, we provide
the optimal transmission policy for each channel region, and the
detailed derivations are given in Appendix~\ref{app:3phase}.

\subsubsection{Region $\mathbb R_{1}(\gamma_{1}\leq\gamma_{3},\gamma_{2}\leq\gamma_{3})$} In this case,
the achievable rate pair $R_{A}$ and $R_{B}$ satisfy (8) and (10),
respectively, which means that the bidirectional links follow direct
transmission. It is obvious that the optimal rates are exactly the
capacity bound, given by
\begin{eqnarray}
R_{A}(\boldsymbol{\epsilon})=\frac{1}{3}C\big(\gamma_{3}P_{A}(\boldsymbol{\epsilon})\big),\label{eqn:e42}\\
R_{B}(\boldsymbol{\epsilon})=\frac{1}{3}C\big(\gamma_{3}P_{B}(\boldsymbol{\epsilon})\big).\label{eqn:e43}
\end{eqnarray}
Substituting the above into the Lagrangian function (\ref{eqn:e30})
to eliminate the rate variables, we can then obtain the closed-form
expressions of the optimal power allocation as
\begin{eqnarray}\label{eqn:e14}
\begin{cases}
P_{A}(\boldsymbol{\epsilon})=\bigg[\bigg(\frac{\sigma\lambda_{A}\phi_{1}}{\omega_{A}}\bigg)^{-\frac{3}{\beta_{A}+3}}\gamma_{3}^{-\frac{\beta_{A}}{\beta_{A}+3}}-\gamma_{3}^{-1}\bigg]^{+}\\
P_{B}(\boldsymbol{\epsilon})=\bigg[\bigg(\frac{\sigma\lambda_{B}\phi_{2}}{\omega_{B}}\bigg)^{-\frac{3}{\beta_{B}+3}}\gamma_{3}^{-\frac{\beta_{B}}{\beta_{B}+3}}-\gamma_{3}^{-1}\bigg]^{+}\\
P_{R}(\boldsymbol{\epsilon})=0,
\end{cases}
\end{eqnarray}
where $\sigma=3\ln 2$, $\beta_{i}=\frac{\theta_{i}}{\ln 2}, i \in
\{A,B\}$, and
\begin{equation}
\phi_{1}=\mathbb{E}_{\boldsymbol\gamma}\big[e^{-\theta_{A}R_{A}(\boldsymbol{\epsilon})}\big],
\end{equation}
\begin{equation}
\phi_{2}=\mathbb{E}_{\boldsymbol\gamma}\big[e^{-\theta_{B}R_{B}(\boldsymbol{\epsilon})}\big].
\end{equation}
Note that $\phi_{1}$, $\phi_{2}$ are expectations over all the
regions, and should be updated with $\boldsymbol\lambda$ in each
iteration of the dual problem.

\subsubsection{Region $\mathbb R_{2}(\gamma_{1}>\gamma_{3},
\gamma_{2}\leq\gamma_{3})$} In this region, node $A$ transmits
signals with the help of node $R$ while node $B$ adopts direct
transmission. According to (7) and (10), the optimal rate allocation
is given by
\begin{equation}\label{eqn:e31}
R_{A}(\boldsymbol{\epsilon})=\frac{1}{3}\min\big\{C\big(\gamma_{1}P_{A}(\boldsymbol{\epsilon})\big),C\big(\gamma_{3}P_{A}(\boldsymbol{\epsilon})\big)+C\big(\gamma_{2}P_{R}(\boldsymbol{\epsilon})\big)\big\},
\end{equation}

\begin{equation}\label{eqn:e44}
R_{B}(\boldsymbol{\epsilon})=\frac{1}{3}C\big(\gamma_{3}P_{B}(\boldsymbol{\epsilon})\big).
\end{equation}
Then the optimal power allocation can be obtained as
\begin{eqnarray}\label{eqn:e15}
\begin{cases}
P_{A}(\boldsymbol{\epsilon})=\big[\widehat{P}_{A}(\boldsymbol{\epsilon})\big]^{+}\\
P_{B}(\boldsymbol{\epsilon})=\bigg[\bigg(\frac{\sigma\lambda_{B}\phi_{2}}{\omega_{B}}\bigg)^{-\frac{3}{\beta_{B}+3}}\gamma_{3}^{-\frac{\beta_{B}}{\beta_{B}+3}}-\gamma_{3}^{-1}\bigg]^{+}\\
P_{R}(\boldsymbol{\epsilon})=\frac{(\gamma_{1}-\gamma_{3})P_{A}(\boldsymbol{\epsilon})}{\gamma_{2}[1+\gamma_{3}P_{A}(\boldsymbol{\epsilon})]},
\end{cases}
\end{eqnarray}
where $\widehat{P}_{A}(\boldsymbol{\epsilon})$ is the solution of
the following equation
\begin{equation}\label{eqn:e16}
\frac{\omega_{A}\gamma_{1}}{\sigma\phi_{1}}\big[1+\gamma_{1}\widehat{P}_{A}(\boldsymbol{\epsilon})\big]^{-\frac{\beta_{A}+3}{3}}-\frac{\lambda_{R}(\gamma_{1}-\gamma_{3})}{\gamma_{2}\big[1+\gamma_{3}\widehat{P}_{A}(\boldsymbol{\epsilon})\big]^{2}}-\lambda_{A}=0.
\end{equation}
We use the simple bisection method to obtain
$\widehat{P}_{A}(\boldsymbol{\epsilon})$ since (\ref{eqn:e16}) is a
monotonically decreasing function of
$\widehat{P}_{A}(\boldsymbol{\epsilon})$.

\subsubsection{Region $\mathbb R_{3}(\gamma_{1}\leq\gamma_{3}, \gamma_{2}>\gamma_{3})$} Node $B$ transmits signals via the assistance of
the relay while node $A$ adopts direct transmission. This case is
similar to that of $\mathbb R_{2}$ and we omit the results here.


\subsubsection{Region $\mathbb R_{4}(\gamma_{1}>\gamma_{3}, \gamma_{2}>\gamma_{3})$} In this region, both source nodes
need the relay node's help. According to the achievable rate pair
(7) and (9), the optimal rate should follow
\begin{equation}\label{eqn:e32}
R_{A}(\boldsymbol{\epsilon})=\frac{1}{3}\min\big\{C\big(\gamma_{1}P_{A}(\boldsymbol{\epsilon})\big),C\big(\gamma_{3}P_{A}(\boldsymbol{\epsilon})\big)+C\big(\gamma_{2}P_{R}(\boldsymbol{\epsilon})\big)\big\},
\end{equation}
\begin{equation}\label{eqn:e33}
R_{B}(\boldsymbol{\epsilon})=\frac{1}{3}\min\big\{C\big(\gamma_{2}P_{B}(\boldsymbol{\epsilon})\big),C\big(\gamma_{3}P_{B}(\boldsymbol{\epsilon})\big)+C\big(\gamma_{1}P_{R}(\boldsymbol{\epsilon})\big)\big\}.
\end{equation}

The associated optimal power allocation is as follows. Define
\begin{equation}\label{eqn:tau}
\tau\triangleq\frac{\gamma_{1}(\gamma_{1}-\gamma_{3})}{\gamma_{2}(\gamma_{2}-\gamma_{3})}.
\end{equation}

If $\tau\leq1$, i.e., $\gamma_{3}<\gamma_{1}\leq\gamma_{2}$, we have
\begin{eqnarray}\label{eqn:e17}
\begin{cases}
P_{A}(\boldsymbol{\epsilon})=\big[\widehat{P}_{A}(\boldsymbol{\epsilon})\big]^{+}\\
P_{B}(\boldsymbol{\epsilon})=\frac{\tau P_{A}(\boldsymbol{\epsilon})}{1+(1-\tau)\gamma_{3}P_{A}(\boldsymbol{\epsilon})}\\
P_{R}(\boldsymbol{\epsilon})=\frac{(\gamma_{1}-\gamma_{3})P_{A}(\boldsymbol{\epsilon})}{\gamma_{2}[1+\gamma_{3}P_{A}(\boldsymbol{\epsilon})]},\\
\end{cases}
\end{eqnarray}
where $\widehat{P}_{A}(\boldsymbol{\epsilon})$ is the solution of
\begin{eqnarray}\label{eqn:e18}
&&\frac{\omega_{A}\gamma_{1}}{\sigma\phi_{1}}\big[1+\gamma_{1}\widehat{P}_{A}(\boldsymbol{\epsilon})\big]^{-\frac{\beta_{A}+3}{3}}+\frac{\tau\omega_{B}\gamma_{2}}{\sigma\phi_{2}\big[1+(1-\tau)\gamma_{3}\widehat{P}_{A}(\boldsymbol{\epsilon})\big]^{2}}\nonumber\\
&&\times\bigg[1+\frac{\tau\gamma_{2}\widehat{P}_{A}(\boldsymbol{\epsilon})}{1+(1-\tau)\gamma_{3}\widehat{P}_{A}(\boldsymbol{\epsilon})}\bigg]^{-\frac{\beta_{B}+3}{3}}-\frac{\lambda_{R}(\gamma_{1}-\gamma_{3})}{\gamma_{2}\big[1+\gamma_{3}\widehat{P}_{A}(\boldsymbol{\epsilon})\big]^{2}}\nonumber\\
&&-\frac{\tau\lambda_{B}}{\big[1+(1-\tau)\gamma_{3}\widehat{P}_{A}(\boldsymbol{\epsilon})\big]^{2}}-\lambda_{A}=0.
\end{eqnarray}
Note that the bisection method can be used to obtain
$\widehat{P}_{A}(\boldsymbol{\epsilon})$ since (\ref{eqn:e18}) is a
monotonically decreasing function of
$\widehat{P}_{A}(\boldsymbol{\epsilon})$.

If $\tau>1$, i.e., $\gamma_{3}<\gamma_{2}<\gamma_{1}$, we have
\begin{eqnarray}\label{eqn:e19}
\begin{cases}
P_{A}(\boldsymbol{\epsilon})=\frac{P_{B}(\boldsymbol{\epsilon})}{\tau+(\tau-1)\gamma_{3}P_{B}(\boldsymbol{\epsilon})}\\
P_{B}(\boldsymbol{\epsilon})=\big[\widehat{P}_{B}(\boldsymbol{\epsilon})\big]^{+}\\
P_{R}(\boldsymbol{\epsilon})=\frac{(\gamma_{2}-\gamma_{3})P_{B}(\boldsymbol{\epsilon})}{\gamma_{1}[1+\gamma_{3}P_{B}(\boldsymbol{\epsilon})]},\\
\end{cases}
\end{eqnarray}
where $\widehat{P}_{B}(\boldsymbol{\epsilon})$ is the solution of
\begin{eqnarray}
&&\frac{\omega_{B}\gamma_{2}}{\sigma\phi_{2}}\big[1+\gamma_{2}\widehat{P}_{B}(\boldsymbol{\epsilon})\big]^{-\frac{\beta_{B}+3}{3}}+\frac{\tau\omega_{A}\gamma_{1}}{\sigma\phi_{1}\big[\tau+(\tau-1)\gamma_{3}\widehat{P}_{B}(\boldsymbol{\epsilon})\big]^{2}}\nonumber\\
&&\times\bigg[1+\frac{\gamma_{1}\widehat{P}_{B}(\boldsymbol{\epsilon})}{\tau+(\tau-1)\gamma_{3}\widehat{P}_{B}(\boldsymbol{\epsilon})}\bigg]^{-\frac{\beta_{A}+3}{3}}-\frac{\lambda_{R}(\gamma_{2}-\gamma_{3})}{\gamma_{1}\big[1+\gamma_{3}\widehat{P}_{B}(\boldsymbol{\epsilon})\big]^{2}}\nonumber\\
&&-\frac{\tau\lambda_{A}}{\big[\tau+(\tau-1)\gamma_{3}\widehat{P}_{B}(\boldsymbol{\epsilon})\big]^{2}}-\lambda_{B}=0,
\end{eqnarray}
which can also be obtained by the bisection method as
(\ref{eqn:e18}).

In summary, the optimal transmission policy
$\{\boldsymbol{P}(\boldsymbol{\epsilon}),\boldsymbol{R}(\boldsymbol{\epsilon})\}$
for the three-phase two-way relaying is given by (\ref{eqn:e14}),
(\ref{eqn:e42}), (\ref{eqn:e43}) when $\gamma_{1}\leq\gamma_{3},
\gamma_{2}\leq\gamma_{3}$, by (\ref{eqn:e15}), (\ref{eqn:e31}),
(\ref{eqn:e44}) when $\gamma_{1}>\gamma_{3},
\gamma_{2}\leq\gamma_{3}$, by (\ref{eqn:e17}), (\ref{eqn:e32}),
(\ref{eqn:e33}) when $\gamma_{3}<\gamma_{1}\leq\gamma_{2}$, and by
(\ref{eqn:e19}), (\ref{eqn:e32}), (\ref{eqn:e33}) when
$\gamma_{3}<\gamma_{2}<\gamma_{1}$. The detailed derivations are
given in Appendix~\ref{app:3phase}.

\subsection{A Special Case}
As reviewed in Section~II-A, the dynamics of $\theta$ correspond to
diverse delay-QoS requirements. In particular, in our two-way relay
system model, if $\theta_{A}=\theta_{B}\rightarrow 0$, meaning the
services of the two nodes are non-real-time, then the weighted sum
of effective capacity yields the weighted ergodic capacity.

Letting $\theta_{A}=\theta_{B}\rightarrow 0$ in (\ref{eqn:e14}),
(\ref{eqn:e15}), (\ref{eqn:e17}) and (\ref{eqn:e19}), we can obtain
the optimal transmission policy for weighted ergodic capacity
maximization. For example, in region $\mathbb R_{1}$, the optimal
powers are given by
\begin{eqnarray}
\begin{cases}
P_{A}^*(\boldsymbol{\epsilon})=\bigg[\bigg(\frac{\sigma\lambda_{A}}{\omega_{A}}\bigg)^{-1}-\gamma_{3}^{-1}\bigg]^{+}\\
P_{B}^*(\boldsymbol{\epsilon})=\bigg[\bigg(\frac{\sigma\lambda_{B}}{\omega_{B}}\bigg)^{-1}-\gamma_{3}^{-1}\bigg]^{+}\\
P_{R}^*(\boldsymbol{\epsilon})=0,
\end{cases}
\end{eqnarray}
which have the standard form of water-filling. For the rest three
regions, similar results can be obtained.

\section{Optimal Policy for Two-Phase Two-way Relaying}

In this section, in accord with
$\mathcal{R}(\boldsymbol{P}(\boldsymbol{\epsilon}),\boldsymbol{\gamma})$
given in (\ref{eqn:e6}), we present the optimal transmission policy
for the two-phase two-way relaying, as well as the optimal partition
criterion in the MAC phase. Meanwhile, the problem in limiting case
when $\theta_{A}=\theta_{B}=0$ is also analyzed.

\subsection{Optimal Policy}

The rate constraints in the BC phase of this protocol can be equally
rewritten as
\begin{eqnarray}
\frac{1}{\theta_{A}}~e^{-\theta_{A}R_{A}(\boldsymbol{\epsilon})}\geq\frac{1}{\theta_{A}}~e^{-\frac{\theta_{A}}{2}C(\gamma_{2}P_{R}(\boldsymbol{\epsilon}))},\label{eqn:e20}\\
\frac{1}{\theta_{B}}~e^{-\theta_{B}R_{B}(\boldsymbol{\epsilon})}\geq\frac{1}{\theta_{B}}~e^{-\frac{\theta_{B}}{2}C(\gamma_{1}P_{R}(\boldsymbol{\epsilon}))}.\label{eqn:e21}
\end{eqnarray}
By using the Lagrange dual method \cite{Boyd}, we can involve the
two rate constraints into the objective function of \textbf{P1}.
Then, the resulting Lagrangian can be expressed as
\begin{eqnarray}\label{eqn:e22}
&&\mathscr{L}\big(\boldsymbol{P}(\boldsymbol{\epsilon}),\boldsymbol{R}(\boldsymbol{\epsilon}),\boldsymbol{\lambda},\boldsymbol{\mu}\big)\nonumber\\
&=&-\frac{\omega_{A}}{\theta_{A}}
\ln\big(\mathbb{E}_{\boldsymbol\gamma}[e^{-\theta_{A}R_{A}(\boldsymbol{\epsilon})}]\big)-\frac{\omega_{B}}{\theta_{B}}\ln\big(\mathbb{E}_{\gamma}[e^{-\theta_{B}R_{B}(\boldsymbol{\epsilon})}]\big)\nonumber\\
&&+\mathbb{E}_{\boldsymbol\gamma}\left[\frac{\mu_{A}}{\theta_{A}}\bigg(e^{-\theta_{A}R_{A}(\boldsymbol{\epsilon})}-e^{-\frac{\theta_{A}}{2}C(\gamma_{2}P_{R}(\boldsymbol{\epsilon}))}\bigg)\right]\nonumber\\
&&+\mathbb{E}_{\boldsymbol\gamma}\left[\frac{\mu_{B}}{\theta_{B}}\bigg(e^{-\theta_{B}R_{B}(\boldsymbol{\epsilon})}-e^{-\frac{\theta_{B}}{2}C(\gamma_{1}P_{R}(\boldsymbol{\epsilon}))}\bigg)\right] \nonumber\\
&&+\lambda_{A}\big(\overline{P_{A}}-\mathbb{E}_{\boldsymbol\gamma}[P_{A}(\boldsymbol{\epsilon})]\big)+\lambda_{B}\big(\overline{P_{B}}-\mathbb{E}_{\boldsymbol\gamma}[P_{B}(\boldsymbol{\epsilon})]\big)\nonumber\\
&&+\lambda_{R}\big(\overline{P_{R}}-\mathbb{E}_{\boldsymbol\gamma}[P_{R}(\boldsymbol{\epsilon})]\big),
\end{eqnarray}
where $\boldsymbol\mu=[\mu_{A},\mu_{B}]^T$ are the Lagrange
multipliers associated with the rate constraints on
$R_{A}(\boldsymbol{\epsilon})$ and $R_{B}(\boldsymbol{\epsilon})$ in
(\ref{eqn:e20}), (\ref{eqn:e21}),
$\boldsymbol{\lambda}=[\lambda_A,\lambda_B,\lambda_R]^T$ are the
Lagrange multipliers related to the power constraints. As a result,
the dual problem can be stated as
\begin{eqnarray}
\textbf{P3}:~&&\min_{\boldsymbol{\lambda}\succeq0,\boldsymbol{\mu}\succeq0}~\max_{\begin{subarray}{c}\boldsymbol{P}(\boldsymbol{\epsilon})\succeq
0,\\\boldsymbol{R}(\boldsymbol{\epsilon})\in
\mathcal{C}_{MAC}\end{subarray}}~\mathscr{L}\big(\boldsymbol{P}(\boldsymbol{\epsilon}),\boldsymbol{R}(\boldsymbol{\epsilon}),\boldsymbol{\lambda},\boldsymbol{\mu}\big).\nonumber
\end{eqnarray}

Taking a close look at the above dual function, we find that the
optimization of the relay power policy
$P_{R}(\boldsymbol{\epsilon})$ can be decoupled from others.
Therefore, the dual function can be computed by solving the two
subproblems as follows,
\begin{eqnarray}
\hspace{-0.2cm}{\rm
subproblem~1}:\nonumber\\
&&\hspace{-2.8cm}\max_{\begin{subarray}{c}P_{A}(\boldsymbol{\epsilon})\geq 0,\\P_{B}(\boldsymbol{\epsilon})\geq 0,\\R_{A}(\boldsymbol{\epsilon}),R_{B}(\boldsymbol{\epsilon})\end{subarray}}~-\frac{\omega_{A}}{\theta_{A}}\ln\big(\mathbb{E}_{\boldsymbol\gamma}[e^{-\theta_{A}R_{A}(\boldsymbol{\epsilon})}]\big)\nonumber\\
&&\hspace{-0.9cm}-\frac{\omega_{B}}{\theta_{B}}\ln\big(\mathbb{E}_{\gamma}[e^{-\theta_{B}R_{B}(\boldsymbol{\epsilon})}]\big)\nonumber\\
&&\hspace{-0.9cm}+\mathbb{E}_{\boldsymbol\gamma}\left[\frac{\mu_{A}}{\theta_{A}}e^{-\theta_{A}R_{A}(\boldsymbol{\epsilon})}\right]+\mathbb{E}_{\boldsymbol\gamma}\left[\frac{\mu_{B}}{\theta_{B}}e^{-\theta_{B}R_{B}(\boldsymbol{\epsilon})}\right] \nonumber\\
&&\hspace{-0.9cm}+\lambda_{A}\big(\overline{P_{A}}-\mathbb{E}_{\boldsymbol\gamma}[P_{A}(\boldsymbol{\epsilon})]\big)\nonumber\\
&&\hspace{-0.9cm}+\lambda_{B}\big(\overline{P_{B}}-\mathbb{E}_{\boldsymbol\gamma}[P_{B}(\boldsymbol{\epsilon})]\big)\label{eqn:e23}\\
&&\hspace{-2.8cm}~~~~s.t.~~~~~~\big(R_{A}(\boldsymbol{\epsilon}),R_{B}(\boldsymbol{\epsilon})\big)\in
\mathcal{C}_{MAC}\big(P_{A}(\boldsymbol{\epsilon}),P_{B}(\boldsymbol{\epsilon}),\boldsymbol{\gamma}\big),
\end{eqnarray}
and
\begin{eqnarray}
&&{\rm subproblem~2}:\nonumber\\
&&\max_{P_{R}(\boldsymbol{\epsilon})\geq 0}~
-\mathbb{E}_{\boldsymbol\gamma}\bigg[\frac{\mu_{A}}{\theta_{A}}\big(1+\gamma_{2}P_{R}(\boldsymbol{\epsilon})\big)^{-\frac{\beta_{A}}{2}}\bigg]\nonumber\\
&&~~~~~~~~~~-\mathbb{E}_{\boldsymbol\gamma}\bigg[\frac{\mu_{B}}{\theta_{B}}\big(1+\gamma_{1}P_{R}(\boldsymbol{\epsilon})\big)^{-\frac{\beta_{B}}{2}}\bigg]\nonumber\\
&&~~~~~~~~~~+\lambda_{R}\big(\overline{P_{R}}-\mathbb{E}_{\boldsymbol\gamma}[P_{R}(\boldsymbol{\epsilon})]\big),\label{eqn:e24}
\end{eqnarray}
where $\beta_{i}=\frac{\theta_{i}}{\ln 2}, i \in \{A,B\}$, is the
same as the definition in Section IV-A. Then the dual problem
\textbf{P3} can be solved through two nested dual searching loops.
The inner loop searches $\boldsymbol\mu$ for given
$\boldsymbol\lambda$ and the outer loop searches
$\boldsymbol\lambda$ for given $\boldsymbol\mu$, where the values
$\boldsymbol\mu$ and $\boldsymbol\lambda$ can be updated iteratively
using the subgradient method with guaranteed convergence as
\begin{eqnarray}
\mu_{A}^{(i+1)}=\mu_{A}^{(i)}-z_{A}^{(i)}\left(e^{-\theta_{A}R_{A}(\boldsymbol{\epsilon})}-e^{-\frac{\theta_{A}}{2}C(\gamma_{2}P_{R}(\boldsymbol{\epsilon}))}\right),\label{eqn:e27}\\
\mu_{B}^{(i+1)}=\mu_{B}^{(i)}-z_{B}^{(i)}\left(e^{-\theta_{B}R_{B}(\boldsymbol{\epsilon})}-e^{-\frac{\theta_{B}}{2}C(\gamma_{1}P_{R}(\boldsymbol{\epsilon}))}\right),\label{eqn:e28}
\end{eqnarray}
\begin{eqnarray}\label{eqn:e29}
\boldsymbol{\lambda}^{(i+1)}=\boldsymbol{\lambda}^{(i)}-\boldsymbol{s}^{(i)}\big(\boldsymbol{\overline{P}}-\mathbb{E}_{\boldsymbol\gamma}[\boldsymbol{P}(\boldsymbol{\epsilon})]\big),
\end{eqnarray}
where the subscript $i$ denotes the iteration index, and
$z_{A}^{(i)}$, $z_{B}^{(i)}$ and $\boldsymbol{s}^{(i)}$ are the step
sizes designed properly. The overall algorithm is specified later.
In the following, we solve the two subproblems respectively.

\subsubsection{Solution of Subproblem 1}
The first subproblem is relevant to $P_{A}(\boldsymbol{\epsilon})$
and $P_{B}(\boldsymbol{\epsilon})$ in the MAC phase but not
$P_{R}(\boldsymbol{\epsilon})$ in the BC phase. It is essentially a
classical resource allocation problem in MAC channels, though the
objective function is slightly different from the Gaussian MAC
\cite{Tse}. As shown in \cite{Tse}, \emph{successive decoding} is
the optimal strategy for resource allocation in MAC channels.
Motivated by this result, we partition the channel states into two
regions, $\mathbb R'_{1}$ and $\mathbb R'_{2}$, for which an example
is shown in Fig.~\ref{fig:csregion}(b). In region $\mathbb R'_{1}$,
the relay first decodes the signal from node $A$, then subtracts
this decoded signal from the received signal, and then decodes the
signal from node $B$. Inversely, in region $\mathbb R'_{2}$, the
relay decodes the signal from node $B$ first and then the signal
from node $A$. Similar partition method is used in the literature
(e.g., \cite{Miller,Qiao}). In the following, we propose the optimal
power and rate adaptation policy for a given channel partition. The
optimal channel partition method will be derived in the next
subsection.

In region $\mathbb R'_{1}$, the maximum rates are achieved at
\cite{Tse}
\begin{eqnarray}
R_{A}(\boldsymbol{\epsilon})&=&\frac{1}{2}C\bigg(\frac{\gamma_{1}P_{A}(\boldsymbol{\epsilon})}{1+\gamma_{2}P_{B}(\boldsymbol{\epsilon})}\bigg),\label{eqn:e45}\\
R_{B}(\boldsymbol{\epsilon})&=&\frac{1}{2}C\big(\gamma_{2}P_{B}(\boldsymbol{\epsilon})\big),\label{eqn:e46}
\end{eqnarray}
while in $\mathbb R'_{2}$, the maximum rates are achieved at
\begin{eqnarray}
R_{A}(\boldsymbol{\epsilon})&=&\frac{1}{2}C\big(\gamma_{1}P_{A}(\boldsymbol{\epsilon})\big),\label{eqn:e47}\\
R_{B}(\boldsymbol{\epsilon})&=&\frac{1}{2}C\bigg(\frac{\gamma_{2}P_{B}(\boldsymbol{\epsilon})}{1+\gamma_{1}P_{A}(\boldsymbol{\epsilon})}\bigg).\label{eqn:e48}
\end{eqnarray}

By applying the Karush-Kuhn-Tucker (KKT) conditions \cite{Boyd}, the
optimal power adaptation policy for $P_{A}(\boldsymbol{\epsilon})$
and $P_{B}(\boldsymbol{\epsilon})$ in $\mathbb R'_{1}$ must satisfy
the following conditions (the derivation is provided in
Appendix~\ref{app:power}):
\begin{eqnarray}\label{eqn:pa}
\begin{cases}
P_{A}(\boldsymbol{\epsilon})=\bigg[\alpha_{1}^{-\frac{2}{\beta_{A}+2}}\bigg(\frac{1+\gamma_{2}P_{B}(\boldsymbol{\epsilon})}{\gamma_{1}}\bigg)^{\frac{\beta_{A}}{\beta_{A}+2}}-\frac{1+\gamma_{2}P_{B}(\boldsymbol{\epsilon})}{\gamma_{1}}\bigg]^{+}\\
P_{B}(\boldsymbol{\epsilon})=\big[\widehat{P}_{B}(\boldsymbol{\epsilon})\big]^{+},
\end{cases}
\end{eqnarray}
where $\widehat{P}_{B}(\boldsymbol{\epsilon})$ can be obtained using
a numerical search\footnote{The outline of this numerical search is:
First, find out the stationary point based on the derivative of the
function. Then, determine the interval where the zero point exists,
and search for this point using the bisection method.} through the
following equation
\begin{eqnarray}\label{eqn:pb}
&&\hspace{-0.3cm}\frac{\gamma_{2}}{\alpha_{2}}\big[1+\gamma_{2}\widehat{P}_{B}(\boldsymbol{\epsilon})\big]^{-\frac{\beta_{B}+2}{2}}+\frac{\lambda_{A}\gamma_{2}}{\lambda_{B}\gamma_{1}}\nonumber\\
&&\hspace{-0.3cm}-\frac{\lambda_{A}\gamma_{2}}{\lambda_{B}\gamma_{1}}\bigg(\frac{\alpha_{1}}{\gamma_{1}}\bigg)^{-\frac{2}{\beta_{A}+2}}\big[1+\gamma_{2}\widehat{P}_{B}(\boldsymbol{\epsilon})\big]^{-\frac{2}{\beta_{A}+2}}-1=0,
\end{eqnarray}
with
\begin{equation}
\alpha_{1}=\frac{\delta\lambda_{A}}{\omega_{A}{\phi_{1}'}^{-1}\nonumber
-\mu_{A}},~~~~
\alpha_{2}=\frac{\delta\lambda_{B}}{\omega_{B}{\phi_{2}'}^{-1}\nonumber
-\mu_{B}},
\end{equation}
\begin{eqnarray}
\phi'_{1}&=&\int_{\boldsymbol{\gamma}\in \mathbb
R'_{1}}\left[1+\frac{\gamma_{1}P_{A}(\boldsymbol{\epsilon})}{1+\gamma_{2}P_{B}(\boldsymbol{\epsilon})}\right]^{-\frac{\beta_{A}}{2}}\mathit
p_{\boldsymbol{\gamma}}(\gamma_{1},\gamma_{2})d\gamma_{1}d\gamma_{2}\nonumber\\
&&+\int_{\boldsymbol{\gamma}\in \mathbb
R'_{2}}\big[1+\gamma_{1}P_{A}(\boldsymbol{\epsilon})\big]^{-\frac{\beta_{A}}{2}}\mathit
p_{\boldsymbol{\gamma}}(\gamma_{1},\gamma_{2})d\gamma_{1}d\gamma_{2},\nonumber
\end{eqnarray}
\begin{eqnarray}
&&\hspace{-0.4cm}\phi'_{2}=\int_{\boldsymbol{\gamma}\in \mathbb
R'_{1}}\big[1+\gamma_{2}P_{B}(\boldsymbol{\epsilon})\big]^{-\frac{\beta_{B}}{2}}\mathit
p_{\boldsymbol{\gamma}}(\gamma_{1},\gamma_{2})d\gamma_{1}d\gamma_{2}\nonumber\\
&&+\int_{\boldsymbol{\gamma}\in \mathbb
R'_{2}}\left[1+\frac{\gamma_{2}P_{B}(\boldsymbol{\epsilon})}{1+\gamma_{1}P_{A}(\boldsymbol{\epsilon})}\right]^{-\frac{\beta_{B}}{2}}\mathit
p_{\boldsymbol{\gamma}}(\gamma_{1},\gamma_{2})d\gamma_{1}d\gamma_{2},\nonumber
\end{eqnarray}
where $\delta=2\ln 2$. Like the three-phase two-way relaying,
$\phi_1'$ and $\phi_2'$ are also updated with $\boldsymbol\lambda$
in each iteration.

Using the similar method, the optimal power allocation in region
$\mathbb R'_{2}$ can be obtained and the details are omitted.

\subsubsection{Solution of Subproblem 2}
The second subproblem is only relevant to
$P_{R}(\boldsymbol{\epsilon})$ in the BC phase, which is also a
convex problem. By applying the KKT conditions, we can obtain the
optimal power allocation
$P_{R}(\boldsymbol{\epsilon})=\big[\widehat{P}_{R}(\boldsymbol{\epsilon})\big]^+$,
where $\widehat{P}_{R}(\boldsymbol{\epsilon})$ must satisfy the
following equation
\begin{eqnarray}\label{eqn:pr}
&&\mu_{A}\gamma_{2}\big[1+\gamma_{2}\widehat{P}_{R}(\boldsymbol{\epsilon})\big]^{-\frac{\beta_{A}+2}{2}}\nonumber\\
&&+\mu_{B}\gamma_{1}\big[1+\gamma_{1}\widehat{P}_{R}(\boldsymbol{\epsilon})\big]^{-\frac{\beta_{B}+2}{2}}-\delta\lambda_{R}=0.
\end{eqnarray}
We apply the bisection method to obtain
$\widehat{P}_{R}(\boldsymbol{\epsilon})$ since (\ref{eqn:pr}) is a
monotonically decreasing function of
$\widehat{P}_{R}(\boldsymbol{\epsilon})$.

\subsection{Optimal Partition Criterion}

In the above subsection, we have demonstrated the optimal
transmission policy for given decoding region $\mathbb R'=(\mathbb
R'_{1},\mathbb R'_{2})$ in the MAC phase. Here we present the
optimal partition criterion to determine the decoding order in the
MAC phase based on the obtained transmission policy.

To maximize the weighted sum effective capacity in (\ref{eqn:e23}),
the optimal decoding order in the MAC phase should be dynamic with
respect to different channel state information. Similar to
\cite{Miller,Qiao}, finding such optimal partition can be viewed as
finding an optimal threshold, $\gamma_{1}^{th}$ for $\gamma_{1}$ or
$\gamma_{2}^{th}$ for $\gamma_{2}$. Here $\gamma_{2}^{th}$ (or
$\gamma_{1}^{th}$) is a function of both $\gamma_{1}$ (or
$\gamma_{2}$) and the power allocation policy
$\boldsymbol{P}(\boldsymbol{\epsilon})$. The following proposition
is established to find the optimal threshold.

\begin{proposition}[Optimal Channel Partition Criterion]
When $\gamma_{2}^{th}$ is used to partition $\mathbb R'$,
$\boldsymbol\gamma$ falls into region $\mathbb R'_{1}$ if
$\gamma_2<\gamma_2^{th}$, otherwise $\boldsymbol\gamma$ falls into
region $\mathbb R'_{2}$, where $\gamma_{2}^{th}$ must
satisfy\footnote{We use the similar numerical method as described
before.}
\begin{equation}\label{eqn:e25}
\frac{\big[1+\frac{\gamma_{1}P_{A}(\boldsymbol{\epsilon})}{1+\gamma_{2}^{th}P_{B}(\boldsymbol{\epsilon})}\big]^{-\frac{\beta_{A}}{2}}-\big[1+\gamma_{1}P_{A}(\boldsymbol{\epsilon})\big]^{-\frac{\beta_{A}}{2}}}
{\big[1+\frac{\gamma_{2}^{th}P_{B}(\boldsymbol{\epsilon})}{1+\gamma_{1}P_{A}(\boldsymbol{\epsilon})}\big]^{-\frac{\beta_{B}}{2}}-\big[1+\gamma_{2}^{th}P_{B}(\boldsymbol{\epsilon})\big]^{-\frac{\beta_{B}}{2}}}=K.
\end{equation}
When $\gamma_{1}^{th}$ is used to partition $\mathbb R'$,
$\boldsymbol\gamma$ falls into region $\mathbb R'_{1}$ if
$\gamma_1>\gamma_1^{th}$, otherwise $\boldsymbol\gamma$ falls into
region $\mathbb R'_{2}$, where $\gamma_{1}^{th}$ must satisfy
\begin{equation}\label{eqn:e26}
\frac{\big[1+\frac{\gamma_{1}^{th}P_{A}(\boldsymbol{\epsilon})}{1+\gamma_{2}P_{B}(\boldsymbol{\epsilon})}\big]^{-\frac{\beta_{A}}{2}}-\big[1+\gamma_{1}^{th}P_{A}(\boldsymbol{\epsilon})\big]^{-\frac{\beta_{A}}{2}}}
{\big[1+\frac{\gamma_{2}P_{B}(\boldsymbol{\epsilon})}{1+\gamma_{1}^{th}P_{A}(\boldsymbol{\epsilon})}\big]^{-\frac{\beta_{B}}{2}}-\big[1+\gamma_{2}P_{B}(\boldsymbol{\epsilon})\big]^{-\frac{\beta_{B}}{2}}}=K.
\end{equation}
In both (\ref{eqn:e25}) and (\ref{eqn:e26}), $K$ is defined as
\begin{equation}
K\triangleq\frac{\theta_{A}(\omega_{B}{\phi'_{2}}^{-1}-\mu_{B})}
{\theta_{B}(\omega_{A}{\phi'_{1}}^{-1}-\mu_{A})}.\nonumber
\end{equation}
\end{proposition}

\begin{proof}
Please see Appendix~\ref{app:partition}.
\end{proof}

Particularly, $\gamma_1^{th}$ and $\gamma_2^{th}$ should be
well-defined in the partition criterion, meaning that they should be
positive. However, as we know from (\ref{eqn:e25}) and
(\ref{eqn:e26}), they cannot be both positive in some condition. In
this case, the obtained positive one is chosen for the partition.

Finally, we describe the overall algorithm in Algorithm~1 to find
the optimal power and rate adaptation policy for the two-phase
two-way relay protocol. Note that in Algorithm 1, for a given
channel partition we can obtain the optimal power allocation, then
the optimal power allocation in turn leads to an optimal channel
partition. Due to the convexity of the problem, the global
convergence and optimality can be guaranteed.

\subsection{A Special Case}

Similar to the three-phase scheme, when
$\theta_{A}=\theta_{B}\rightarrow 0$,  we can obtain the optimal
transmission policy for the two-phase two-way relaying without delay
requirements (i.e., the ergodic capacity).

\begin{proposition}
The optimal power allocation policy for the two-phase two-way DF
relaying for weighted ergodic capacity maximization when
$\xi_{A}<\xi_{B}$ is determined by
\begin{equation}\label{eqn:e49}
P_{A}^*(\boldsymbol{\epsilon})=\bigg[\frac{\xi_{A}}{\delta
\lambda_{A}}-\frac{\xi_{B}-\xi_{A}}{\delta\gamma_{1}
(\frac{\lambda_{B}}{\gamma_{2}}-\frac{\lambda_{A}}{\gamma_{1}})}\bigg]^{+},
\end{equation}
\begin{equation}\label{eqn:e50}
P_{B}^*(\boldsymbol{\epsilon})=\bigg[\frac{\xi_{B}-\xi_{A}}{\delta\gamma_{2}
(\frac{\lambda_{B}}{\gamma_{2}}-\frac{\lambda_{A}}{\gamma_{1}})}-\frac{1}{\gamma_{2}}\bigg]^{+},
\end{equation}
\begin{eqnarray}\label{eqn:e51}
P_{R}^*(\boldsymbol{\epsilon})=\begin{cases} 0, &
\lambda_{R}\geq\frac{\mu_{A}\gamma_{1}+\mu_{B}\gamma_{2}}{\delta}\\
\frac{-c_{2}+ \sqrt{{c_{2}^{2}-4c_{1}c_{3}}}}{2c_{1}}, & {\rm
otherwise}
\end{cases}
\end{eqnarray}
where $\xi_{A}=\omega_{A}-\mu_{A}$, $\xi_{B}=\omega_{B}-\mu_{B}$,
and $c_{1}=\lambda_{R}\gamma_{1}\gamma_{2}$,
$c_{2}=\lambda_{R}(\gamma_{1}+\gamma_{2})-\gamma_{1}\gamma_{2}(\mu_{A}+\mu_{B})/\delta$,
$c_{3}=\lambda_{R}-(\mu_{A}\gamma_{1}+\mu_{B}\gamma_{2})/\delta$.
The results for $\xi_{A}\geq\xi_{B}$ can be easily obtained using
the same methods.
\end{proposition}

\begin{proof} Letting $\theta_{A}=\theta_{B}\rightarrow 0$ in (\ref{eqn:pa}) and (\ref{eqn:pr}),
we have
\begin{eqnarray}
\begin{cases}
P_{A}(\boldsymbol{\epsilon})=\big[\frac{\xi_{A}}{\delta\lambda_{A}}-\frac{1+\gamma_{2}P_{B}(\boldsymbol{\epsilon})}{\gamma_{1}}\big]^{+}\\
P_{B}(\boldsymbol{\epsilon})=\big[\frac{\gamma_{2}(\xi_{B}-\xi_{A})}{\delta(\lambda_{B}\gamma_{1}-\lambda_{A}\gamma_{2})}-\frac{1}{\gamma_{2}}\big]^{+}\\
\frac{\mu_{A}\gamma_{2}}{1+\gamma_{2}P_{R}(\boldsymbol{\epsilon})}+\frac{\mu_{B}\gamma_{1}}{1+\gamma_{1}P_{R}(\boldsymbol{\epsilon})}-\delta\lambda_{R}=0.\nonumber
\end{cases}
\end{eqnarray}
Through simple calculation, we can get the desired results
(\ref{eqn:e49}), (\ref{eqn:e50}) and (\ref{eqn:e51}) in
$\mathbb{R}'_1$.

For the segmentation of $(\mathbb R'_{1},\mathbb R'_{2})$ in the
case of $\theta_{A}=\theta_{B}\rightarrow 0$, by applying the
proposed optimal partition criterion, we have
\begin{equation}
K=\frac{\theta_{A}(\omega_{B}-\mu_{B})}{\theta_{B}(\omega_{A}-\mu_{A})}=\frac{\xi_{B}}{\xi_{A}}.\nonumber
\end{equation}
Under this condition, (\ref{eqn:e25}) becomes
\begin{equation}
\bigg[\frac{1+\gamma_{1}P_{A}(\boldsymbol{\epsilon})}{1+\gamma_{2}^{th}P_{B}(\boldsymbol{\epsilon})}\bigg]^{-\frac{\beta_{A}}{2}}=K.\nonumber
\end{equation}
When $\xi_{B}>\xi_{A}$, meaning that $K>1$, $\gamma_{2}^{th}>0$ is
well-defined for partition. As $\theta_{A} \rightarrow 0$, we can
find
\begin{equation}
\gamma_{2}^{th}\rightarrow\infty, \nonumber
\end{equation}
which implies that $\boldsymbol\gamma$ always falls into
$\mathbb{R}'_1$. Hence, the proof completes. Similar analysis can
also be done when $0 \leq K \leq 1$ by using $\gamma_{1}^{th}$.
\end{proof}

\begin{algorithm}[!t]
\caption{Finding the optimal transmission policy for two-phase
two-way relaying}
\begin{algorithmic}[1]
\STATE Given $\omega_{A}$, $\omega_{B}$. \STATE \textbf{Initialize}
$P_{A}(\boldsymbol{\epsilon})=\overline{P_{A}}$,
$P_{B}(\boldsymbol{\epsilon})=\overline{P_{B}}$. \STATE
\textbf{Initialize} $\boldsymbol\lambda$. \REPEAT \STATE
\textbf{Initialize} $\boldsymbol\mu$ for each $\boldsymbol\gamma$.
\REPEAT \STATE Determine the decoding order in the MAC phase
according to the optimal partition criterion given in Proposition 1.
\IF{the signal from node $A$ is decoded first, i.e.,
$\boldsymbol\gamma\in\mathbb{R}_1'$} \STATE Obtain the two source
power allocations $P_{A}(\boldsymbol{\epsilon})$ and
$P_{B}(\boldsymbol{\epsilon})$ by using (\ref{eqn:pa}). \STATE
Obtain the relay power allocation $P_{R}(\boldsymbol{\epsilon})$ by
using (\ref{eqn:pr}). \STATE Obtain the rate adaptation
$\boldsymbol{R}(\boldsymbol{\epsilon})$ as (\ref{eqn:e45}),
(\ref{eqn:e46}). \STATE Update $\boldsymbol\mu$ using the
subgradient method in (\ref{eqn:e27}), (\ref{eqn:e28}). \ELSE \STATE
$\boldsymbol\gamma\in\mathbb{R}_2'$, adopt similar operations as
above except that the rate adaptation is given in (\ref{eqn:e47}),
(\ref{eqn:e48}). \ENDIF \UNTIL{$\boldsymbol\mu$ converges.} \STATE
Update $\boldsymbol\lambda$ using the subgradient method in
(\ref{eqn:e29}). \UNTIL{$\boldsymbol\lambda$ converges.}
\end{algorithmic}
\end{algorithm}

\section{Numerical Results}

In this section, extensive numerical results are provided to
illustrate the performance of our proposed cross-layer transmission
strategies for the two-way relay systems.

As a benchmark, the conventional two-way channel with direct
transmission is considered, for which the optimal transmission
policy is obtained from \cite{Jia}. Besides, to show the advantage
of the optimal channel partition in the two-phase two-way relaying,
the transmission policy using the static weight-based channel
partition, as introduced in \cite{Tse}, is studied. Specifically, in
the MAC phase of this scheme, the relay always first decodes the
signals from the source with smaller weight, regardless of the CSI
variation. The weight-based and the proposed CSI-based schemes with
fixed power assignment are also studied to illustrate the
significance of channel-aware power adaptation. In these fixed power
assignment schemes, the instantaneous power of all transmitting
nodes is set to be a constant, while the transmission rates are
adaptive with respect to the channel fading.

In our numerical evaluation, the relay is located in a line between
the two users. We set the distance between node $A$ and $B$ as $2$.
The $A$-$R$ distance and the $B$-$R$ distance are denoted as $d$ and
$2-d$, respectively, where $0<d<2$. The log-distance path loss model
with small-scale Rayleigh fading is assumed. Hence, the network
channel information $\gamma_{1}, \gamma_{2}, \gamma_{3}$ follow
independent exponential distribution with parameter
$\lambda_{1}=d^{\nu}$, $\lambda_{2}=(2-d)^{\nu}$, and
$\lambda_{3}=2^{\nu}$, respectively, where $\nu$ denotes the path
loss exponent. A typical value of $\nu$ lies in the range of $(2,
5)$, and it is set as $4$ in our examples. The long-term power
constraints of the nodes satisfy
$\overline{P_{A}}=\overline{P_{B}}=\overline{P_{R}}+3{\rm dB}$.
Throughout this section, the weights are given by $\omega_{A}=0.6$
and $\omega_{B}=0.4$ for illustration purpose.

\subsection{Performance of Symmetric Relay for Two Sources}

In this subsection, we consider the case when relay is in the middle
of the two source nodes, namely, $d=1$. Hence, the channels between
the sources and the relay are symmetric.

\begin{figure}[!t]
\begin{centering}
\includegraphics[scale=.57]{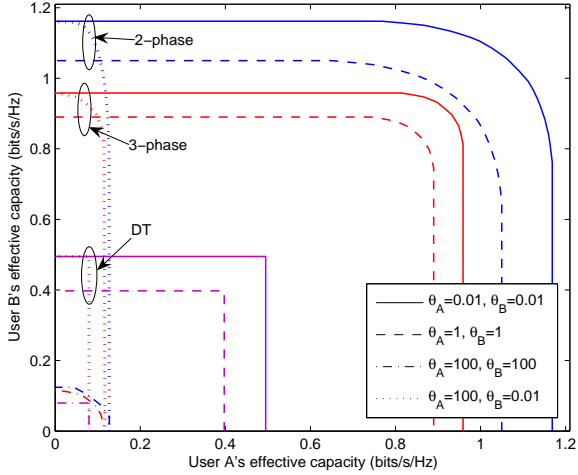}
\vspace{-0.01cm}
 \caption{Optimal effective capacity regions for different schemes under diverse delay constraints with $\overline{P_{A}}=\overline{P_{B}}=9{\rm dB},\overline{P_{R}}=6{\rm dB}$.}\label{fig:region_dif}
\end{centering}
\vspace{-0.3cm}
\end{figure}
\begin{figure}[!t]
\begin{centering}
\includegraphics[scale=.55]{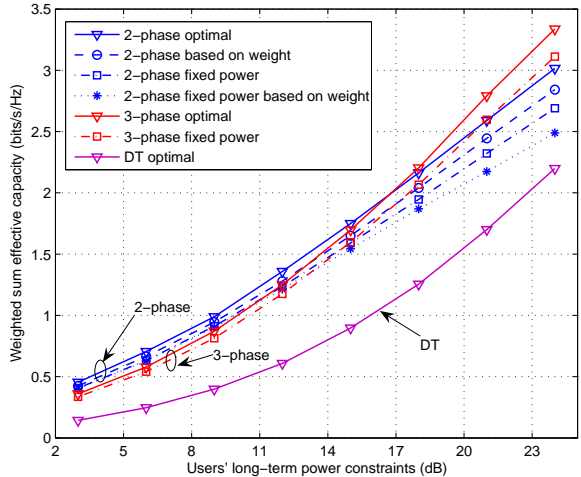}
\vspace{-0.1cm}
 \caption{Weighted sum effective capacity versus long-term power constraints under $\theta_{A}=\theta_{B}=1$.}\label{fig:power}
\end{centering}
\vspace{-0.3cm}
\end{figure}

Firstly, in Fig.~\ref{fig:region_dif}, we plot the optimal effective
capacity regions under different delay constraints for the
three-phase and two-phase two-way relaying, as well as the two-way
direct transmission. We set the long-term power constraints of the
source nodes as $9{\rm dB}$. It is observed from
Fig.~\ref{fig:region_dif} that, with the help of two-way relay, the
effective capacity region is significantly expanded compared with
the conventional two-way direct transmission. We can also find that,
if one user's delay constraint becomes stringent, its effective
capacity becomes small, so is the effective capacity region. This
suggests that there is in general a fundamental throughput-delay
tradeoff associated with the optimal resource allocation. In
addition, under these given power constraints, the effective
capacity region of the two-phase protocol is larger than that of the
three-phase protocol.

\begin{figure}[!t]
\begin{centering}
\includegraphics[scale=.52]{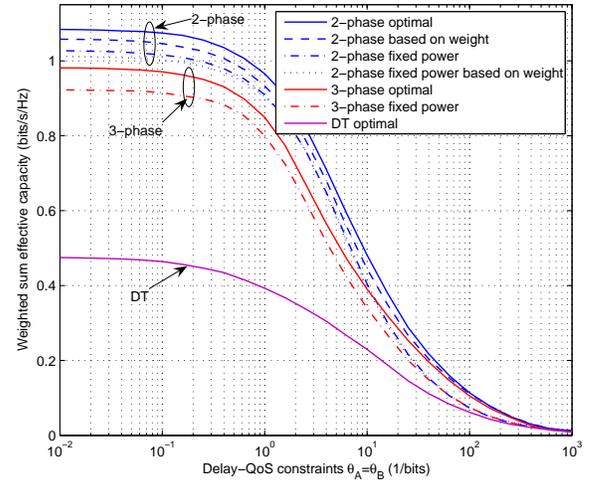}
\vspace{-0.1cm}
 \caption{Weighted sum effective capacity versus different delay constraints with long-term power constraints $\overline{P_{A}}=\overline{P_{B}}=9{\rm dB},\overline{P_{R}}=6{\rm dB}$.}\label{fig:theta}
\end{centering}
\vspace{-0.3cm}
\end{figure}
\begin{figure}[!t]
\begin{centering}
\includegraphics[scale=.63]{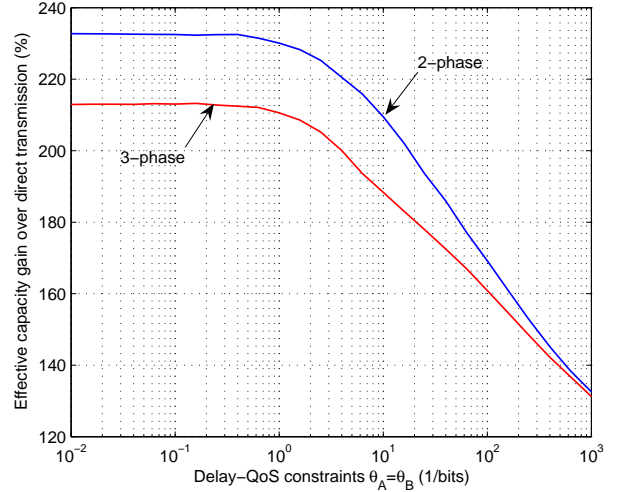}
\vspace{-0.1cm}
 \caption{Effective capacity gain of optimal policies over direct transmission with long-term power constraints $\overline{P_{A}}=\overline{P_{B}}=9{\rm dB},\overline{P_{R}}=6{\rm dB}$.}\label{fig:gain}
\end{centering}
\vspace{-0.3cm}
\end{figure}

In Fig.~\ref{fig:power}, we compare the weighted sum effective
capacity of different schemes under different long-term power
constraints, where the delay constraints are set as
$\theta_{A}=\theta_{B}=1$.
From this figure, we can observe that the two-way relaying brings
tremendous improvements on the weighted sum effective capacity for
information exchange between sources. By taking a closer look at
Fig.~\ref{fig:power}, it can be found that the power adaptation can
bring about $10\%$ and $7\%$ improvements on the effective capacity
for the two-phase and three-phase protocols, respectively, at all
the considered power constraints. For the two-phase protocol, about
$5\%$ performance improvement can be achieved by the proposed
CSI-based scheme over the weight-based scheme. Moreover, we can also
see that the weighted sum effective capacity of the two-phase
protocol is superior to that of the three-phase protocol when the
source power is below about $18{\rm dB}$, while it is inferior to
the three-phase protocol when the source power is higher. Therefore,
the three-phase scheme with power adaptation is more appropriate for
cross-layer two-way relaying under high SNR conditions.

Fig.~\ref{fig:theta} shows the weighted sum effective capacity
versus different delay-QoS constraints
$\theta=\theta_{A}=\theta_{B}$, and the effective capacity gain of
the optimal policies over direct transmission are further plotted in
Fig.~\ref{fig:gain}. Here, the power constraints for the source
nodes are fixed as $9{\rm dB}$. As presented in
Fig.~\ref{fig:theta}, the weighted sum effective capacity generally
decreases with the increasing $\theta$. We can observe from both
Fig.~\ref{fig:theta} and Fig.~\ref{fig:gain} that when the delay
constraints are loose, the optimal policies for the two-phase and
three-phase two-way relaying can achieve substantial effective
capacity gains over the direct transmission. However, the advantages
become small as the delay constraints go stringent. Particularly,
the weighted sum effective capacity of all schemes approach to zero
if $\theta$ is large enough. This is expected as, when the delay
constraints are very stringent, the system can no longer support the
transmission due to fading effect of the channel. This conclusion is
consistent with the theory of delay-limited capacity in information
theory. Moreover, it is obvious that the proposed strategies can
efficiently provide the best weighted sum effective capacity in two
different protocols. Specifically, the margins between the optimal
policies and the fixed power policies for both protocols go larger
along with the decrease of $\theta$. Besides, it has demonstrated
that, under such condition, the two-phase protocol is superior to
the three-phase protocol, though the superiority becomes small for
stringent delay constraints.

\begin{figure}[!t]
\begin{centering}
\includegraphics[scale=.55]{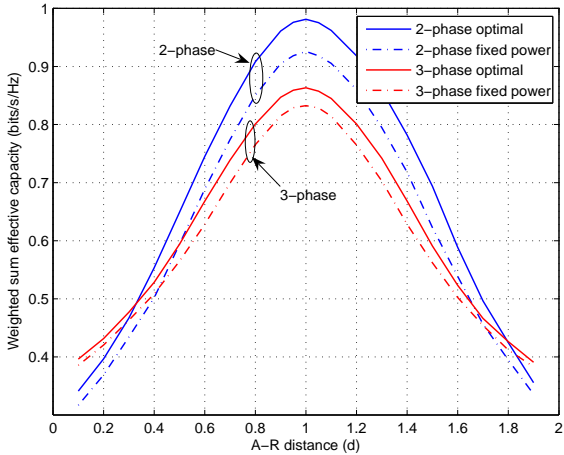}
\vspace{-0.1cm}
 \caption{Impact of relay location on weighted sum effective capacity with same delay constraint $\theta_{A}=\theta_{B}=1$.}\label{fig:d_s}
\end{centering}
\vspace{-0.3cm}
\end{figure}

\subsection{Impact of Relay Location on System Performance}

In this subsection, we will consider the impact of the relay
location in the two-way relay system. Here, we assume the long-term
power constraints of the source nodes are $9{\rm dB}$.

First, we consider the case when source nodes have the same
delay-QoS constraint that $\theta_{A}=\theta_{B}=1$. As illustrated
in Fig.~\ref{fig:d_s}, the maximal weighted sum effective capacity
is obtained when the relay is in the middle of the two source nodes
no matter what transmission strategy is adopted. Meanwhile, we can
observe that our proposed resource allocation policies can obviously
achieve effective capacity gains in both transmission schemes.
However, the benefits decrease when the relay is close to either of
the source nodes. The reason is that, as a result of the severe path
loss, the channel between the relay and the distant source becomes
the major limit of the transmission in this case. Moreover, we can
find that, on this condition, when the distance from the relay to
the source node is less than 0.3, the three-phase protocol
outperforms the two-phase protocol. Otherwise, the two-phase
protocol has advantage over the three-phase protocol.

Next, we study the situation where the two source nodes have
different delay requirements. Here, we set that $\theta_{A}=100$ and
$\theta_{B}=1$. It is interesting to find from Fig.~\ref{fig:d_dif}
that, the maximal weighted sum effective capacity for the
three-phase protocol is gained when the relay is in the middle of
the two sources, while the relay should be closer to the source with
more stringent delay requirement if the two-phase protocol is used.
Different from Fig.~\ref{fig:d_s}, when the distance between the
relay and the node with greater $\theta$ is larger than about 1.3,
the three-phase protocol can get better performance.

\begin{figure}[!t]
\begin{centering}
\includegraphics[scale=.59]{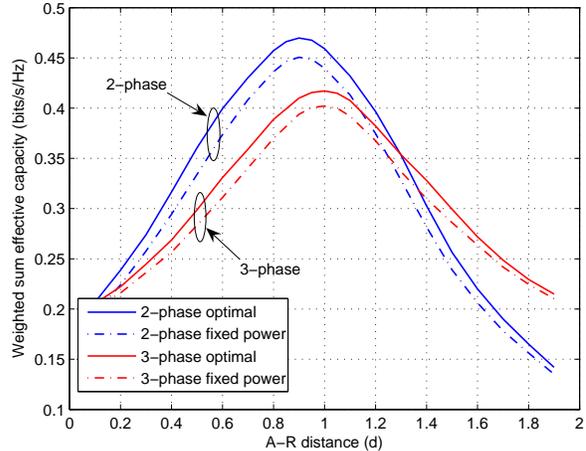}
\vspace{-0.1cm}
 \caption{Impact of relay location on weighted sum effective capacity with delay constraints $\theta_{A}=100, \theta_{B}=1$.}\label{fig:d_dif}
\end{centering}
\vspace{-0.3cm}
\end{figure}

\section{Conclusion}

This paper has studied the cross-layer optimization of two-way
relaying under statistical delay-QoS constraints. We have focused on
two main transmission protocols: three-phase transmission and
two-phase transmission. By integrating the theory of effective
capacity, the optimization problem for weighted sum throughput
maximization in physical layer and delay provisioning in datalink
layer was modeled as a long-term weighted sum effective capacity
maximization problem. Then, optimal transmission policy was proposed
for each protocol.

A few important conclusions have been made through extensive
numerical results. Firstly, our proposed two-way relaying policies
can efficiently provide delay-QoS guarantees, while there exists a
tradeoff between the throughput gain and the delay-QoS provisioning.
Secondly, the proposed policies significantly improve the system
performance compared with the schemes without power adaptation.
Especially, for the two-phase protocol, the proposed CSI-based
method for successive decoding in the MAC phase has 5-10$\%$
performance improvements compared with the weight-based successive
decoding. Thirdly, when the relay is located in the middle of the
transmission, in terms of weighted sum effective capacity, the
three-phase procotol outperforms the two-phase protocol in high SNR
regime and is inferior to the two-phase protocol in low SNR regime.
Last but not least, it is better to place the relay closer to the
source with more stringent delay constraint for the two-phase
protocol.

This work has concentrated on DF relaying with full channel
information in two-way relay systems. It can be further extended to
the heterogeneous networks consisting of delay-constrained and
non-delay-constrained traffics. Alternative relaying protocols and
transmission strategies with time adaptation or partial channel
state information are also possible for future research.

\appendices
\section{Derivation of Optimal Power Adaptation for Three-phase Two-way Relaying}\label{app:3phase}

Substituting the optimal rate assignment at each decoding region
into the Lagrangian (\ref{eqn:e30}), we have
\begin{eqnarray}\label{eqn:3t}
&&-\frac{\omega_{A}}{\theta_{A}}
\ln\bigg(\int_{\boldsymbol{\gamma}\in \mathbb R_{1},\mathbb
R_{3}}\big[1+\gamma_{3}P_{A}(\boldsymbol{\epsilon})\big]^{-\frac{\beta_{A}}{3}}\mathit
p_{\boldsymbol{\gamma}}(\gamma_{1},\gamma_{2},\gamma_{3})\nonumber\\
&&\times
d\gamma_{1}d\gamma_{2}d\gamma_{3}+\int_{\boldsymbol{\gamma}\in
\mathbb R_{2},\mathbb
R_{4}}\max\big\{\big[1+\gamma_{1}P_{A}(\boldsymbol{\epsilon})\big]^{-\frac{\beta_{A}}{3}},\nonumber\\
&&\big[1+\gamma_{3}P_{A}(\boldsymbol{\epsilon})+\gamma_{2}P_{R}(\boldsymbol{\epsilon})+\gamma_{3}P_{A}(\boldsymbol{\epsilon})\gamma_{2}P_{R}(\boldsymbol{\epsilon})\big]^{-\frac{\beta_{A}}{3}}\big\}\nonumber\\
&&\times\mathit
p_{\boldsymbol{\gamma}}(\gamma_{1},\gamma_{2},\gamma_{3})d\gamma_{1}d\gamma_{2}d\gamma_{3}\bigg)\nonumber\\
&&-\frac{\omega_{B}}{\theta_{B}}
\ln\bigg(\int_{\boldsymbol{\gamma}\in \mathbb R_{1},\mathbb
R_{2}}\big[1+\gamma_{3}P_{B}(\boldsymbol{\epsilon})\big]^{-\frac{\beta_{B}}{3}}\mathit
p_{\boldsymbol{\gamma}}(\gamma_{1},\gamma_{2},\gamma_{3})\nonumber\\
&&\times
d\gamma_{1}d\gamma_{2}d\gamma_{3}+\int_{\boldsymbol{\gamma}\in
\mathbb R_{3},\mathbb
R_{4}}\max\big\{\big[1+\gamma_{2}P_{B}(\boldsymbol{\epsilon})\big]^{-\frac{\beta_{B}}{3}},\nonumber\\
&&\big[1+\gamma_{3}P_{B}(\boldsymbol{\epsilon})+\gamma_{1}P_{R}(\boldsymbol{\epsilon})+\gamma_{3}P_{B}(\boldsymbol{\epsilon})\gamma_{1}P_{R}(\boldsymbol{\epsilon})\big]^{-\frac{\beta_{B}}{3}}\big\}\nonumber\\
&&\times\mathit
p_{\boldsymbol{\gamma}}(\gamma_{1},\gamma_{2},\gamma_{3})d\gamma_{1}d\gamma_{2}d\gamma_{3}\bigg)+\lambda_{A}\big(\overline{P_{A}}-\mathbb{E}_{\boldsymbol\gamma}[P_{A}(\boldsymbol{\epsilon})]\big)\nonumber\\
&&+\lambda_{B}\big(\overline{P_{B}}-\mathbb{E}_{\boldsymbol\gamma}[P_{B}(\boldsymbol{\epsilon})]\big)+\lambda_{R}\big(\overline{P_{R}}-\mathbb{E}_{\boldsymbol\gamma}[P_{R}(\boldsymbol{\epsilon})]\big),
\end{eqnarray}
where $p_{\boldsymbol\gamma}$ is the distribution function of
$\boldsymbol\gamma$.

\subsubsection{Region $\mathbb R_{1}(\gamma_{1}\leq\gamma_{3}, \gamma_{2}\leq\gamma_{3})$}Setting the partial derivative at $P_{A}(\boldsymbol{\epsilon})$
equal to zero, we can obtain
\begin{eqnarray}
\int_{\boldsymbol{\gamma}\in \mathbb
R_{1}}&&\bigg\{\frac{\omega_{A}}{\theta_{A}\phi_{1}}\cdot\frac{\beta_{A}\gamma_{3}}{3}\big[1+\gamma_{3}P_{A}(\boldsymbol{\epsilon})\big]^{-\frac{\beta_{A}+3}{3}}-\lambda_{A}\bigg\}\nonumber\\
&&\times\mathit
p_{\boldsymbol{\gamma}}(\gamma_{1},\gamma_{2},\gamma_{3})d\gamma_{1}d\gamma_{2}d\gamma_{3}=0,
\end{eqnarray}
which gives
\begin{equation}
\frac{\omega_{A}\gamma_{3}}{\sigma\phi_{1}}\big[1+\gamma_{3}P_{A}(\boldsymbol{\epsilon})\big]^{-\frac{\beta_{A}+3}{3}}=\lambda_{A}.
\end{equation}
Hence, the optimal $P_{A}(\boldsymbol{\epsilon})$ can be obtained as
in (\ref{eqn:e14}). The optimal $P_{B}(\boldsymbol{\epsilon})$ can
be achieved using the same method.

\subsubsection{Region $\mathbb R_{2}(\gamma_{1}>\gamma_{3}, \gamma_{2}\leq\gamma_{3})$}It is easy to see that the maximum $R_A(\boldsymbol{\epsilon})$ is
achieved when
\begin{equation}
C\big(\gamma_{1}P_{A}(\boldsymbol{\epsilon})\big)=C\big(\gamma_{3}P_{A}(\boldsymbol{\epsilon})\big)+C\big(\gamma_{2}P_{R}(\boldsymbol{\epsilon})\big),
\end{equation}
which in turn gives
\begin{equation}\label{eqn:e34}
P_{R}(\boldsymbol{\epsilon})=\frac{(\gamma_{1}-\gamma_{3})P_{A}(\boldsymbol{\epsilon})}{\gamma_{2}\big[1+\gamma_{3}P_{A}(\boldsymbol{\epsilon})\big]}.
\end{equation}
We substitute (\ref{eqn:e34}) into (\ref{eqn:3t}), take the partial
derivative at $P_{A}(\boldsymbol{\epsilon})$, and then obtain
\begin{eqnarray}
\int_{\boldsymbol{\gamma}\in \mathbb R_{2}}
&&\bigg\{\frac{\omega_{A}}{\theta_{A}\phi_{1}}\cdot\frac{\beta_{A}\gamma_{1}}{3}\big[1+\gamma_{1}P_{A}(\boldsymbol{\epsilon})\big]^{-\frac{\beta_{A}+3}{3}}\nonumber\\
&&-\frac{\lambda_{R}(\gamma_{1}-\gamma_{3})}{\gamma_{2}}\cdot\frac{1}{\big[1+\gamma_{3}P_{A}(\boldsymbol{\epsilon})\big]^{2}}-\lambda_{A}\bigg\}\nonumber\\
&&\times\mathit
p_{\boldsymbol{\gamma}}(\gamma_{1},\gamma_{2},\gamma_{3})d\gamma_{1}d\gamma_{2}d\gamma_{3}=0,
\end{eqnarray}
which yields (\ref{eqn:e16}). Therefore, the optimal
$P_{A}(\boldsymbol{\epsilon})$ is given in (\ref{eqn:e15}). The
power allocation for $P_{B}(\boldsymbol{\epsilon})$ is same as
$\mathbb R_{1}$.

\subsubsection{Region $\mathbb R_{3}(\gamma_{1}\leq\gamma_{3}, \gamma_{2}>\gamma_{3})$}We can use the similar method as
$\mathbb R_{2}$ to get the optimal power allocation for
$P_{A}(\boldsymbol{\epsilon}),P_{B}(\boldsymbol{\epsilon})$.

\subsubsection{Region $\mathbb R_{4}(\gamma_{1}>\gamma_{3}, \gamma_{2}>\gamma_{3})$}According to (7) and (9), the maximum rates are achieved when
\begin{eqnarray}
C\big(\gamma_{1}P_{A}(\boldsymbol{\epsilon})\big)&=&C\big(\gamma_{3}P_{A}(\boldsymbol{\epsilon})\big)+C\big(\gamma_{2}P_{R}(\boldsymbol{\epsilon})\big),
\end{eqnarray}
\begin{eqnarray}
C\big(\gamma_{2}P_{B}(\boldsymbol{\epsilon})\big)&=&C\big(\gamma_{3}P_{B}(\boldsymbol{\epsilon})\big)+C\big(\gamma_{1}P_{R}(\boldsymbol{\epsilon})\big),
\end{eqnarray}
which offers
\begin{equation}\label{eqn:e35}
P_{R}(\boldsymbol{\epsilon})=\frac{(\gamma_{1}-\gamma_{3})P_{A}(\boldsymbol{\epsilon})}{\gamma_{2}\big[1+\gamma_{3}P_{A}(\boldsymbol{\epsilon})\big]}=\frac{(\gamma_{2}-\gamma_{3})P_{B}(\boldsymbol{\epsilon})}{\gamma_{1}\big[1+\gamma_{3}P_{B}(\boldsymbol{\epsilon})\big]}.
\end{equation}
Therefore,
\begin{equation}
P_{A}(\boldsymbol{\epsilon})=\frac{P_{B}(\boldsymbol{\epsilon})}{\tau+(\tau-1)\gamma_{3}P_{B}(\boldsymbol{\epsilon})},
\end{equation}
\begin{equation}\label{eqn:e36}
P_{B}(\boldsymbol{\epsilon})=\frac{\tau
P_{A}(\boldsymbol{\epsilon})}{1+(1-\tau)\gamma_{3}P_{A}(\boldsymbol{\epsilon})},
\end{equation}
where $\tau$ is given in (\ref{eqn:tau}). If $0<\tau\leq1$, it
always holds that if $P_{A}(\boldsymbol{\epsilon})>0$, then
$P_{B}(\boldsymbol{\epsilon})>0$. Otherwise, if $\tau>1$, then if
$P_{B}(\boldsymbol{\epsilon})>0$, it is sure that
$P_{A}(\boldsymbol{\epsilon})>0$. In the following, we assume
$0<\tau\leq1$, while the power allocation policies for $\tau>1$ can
be obtained in the same manner.

Substituting (\ref{eqn:e35}) and (\ref{eqn:e36}) into (\ref{eqn:3t})
and taking the the partial derivative of
$P_{A}(\boldsymbol{\epsilon})$, we can get
\begin{eqnarray}
\int_{\boldsymbol{\gamma}\in \mathbb R_{4}}
&&\bigg\{\frac{\omega_{A}}{\theta_{A}\phi_{1}}\cdot\frac{\beta_{A}\gamma_{1}}{3}\big[1+\gamma_{1}P_{A}(\boldsymbol{\epsilon})\big]^{-\frac{\beta_{A}+3}{3}}\nonumber\\
&&+\frac{\omega_{B}}{\theta_{B}\phi_{2}}\cdot\frac{\beta_{B}}{3}\cdot\frac{\tau\gamma_{2}}{\big[1+(1-\tau)\gamma_{3}P_{A}(\boldsymbol{\epsilon})\big]^{2}}\nonumber\\
&&\times\bigg[1+\frac{\tau
\gamma_{2}P_{A}(\boldsymbol{\epsilon})}{1+(1-\tau)\gamma_{3}P_{A}(\boldsymbol{\epsilon})}\bigg]^{-\frac{\beta_{B}+3}{3}}-\lambda_{A}\nonumber\\
&&-\frac{\tau\lambda_{B}}{\big[1+(1-\tau)\gamma_{3}P_{A}(\boldsymbol{\epsilon})\big]^{2}}-\frac{\lambda_{R}(\gamma_{1}-\gamma_{3})}{\gamma_{2}}\nonumber\\
&&\times\frac{1}{\big[1+\gamma_{3}P_{A}(\boldsymbol{\epsilon})\big]^{2}}\bigg\}\mathit
p_{\boldsymbol{\gamma}}(\gamma_{1},\gamma_{2},\gamma_{3})d\gamma_{1}d\gamma_{2}d\gamma_{3}=0,\nonumber\\
\end{eqnarray}
which can be simplified as equation (\ref{eqn:e18}). Thus, the
optimal $P_{A}$ is the solution of (\ref{eqn:e18}) as stated in
(\ref{eqn:e17}). The optimal $P_{B}$ can be obtained similarly.

\section{Derivation of Optimal $P_A(\boldsymbol{\epsilon})$ and $P_B(\boldsymbol{\epsilon})$ for Two-phase Two-way Relaying}\label{app:power}
According to the achievable rates in the decoding region $(\mathbb
R'_{1},\mathbb R'_{2})$, we can rewrite (\ref{eqn:e23}) without loss
of optimality as
\begin{eqnarray}\label{eqn:mac}
&&-\frac{\omega_{A}}{\theta_{A}}\ln\bigg(\int_{\boldsymbol{\gamma}\in
\mathbb
R'_{1}}\bigg[1+\frac{\gamma_{1}P_{A}(\boldsymbol{\epsilon})}{1+\gamma_{2}P_{B}(\boldsymbol{\epsilon})}\bigg]^{-\frac{\beta_{A}}{2}}\mathit
p_{\boldsymbol{\gamma}}(\gamma_{1},\gamma_{2})d\gamma_{1}d\gamma_{2}\nonumber\\
&&+\int_{\boldsymbol{\gamma}\in \mathbb
R'_{2}}\big[1+\gamma_{1}P_{A}(\boldsymbol{\epsilon})\big]^{-\frac{\beta_{A}}{2}}\mathit
p_{\boldsymbol{\gamma}}(\gamma_{1},\gamma_{2})d\gamma_{1}d\gamma_{2}\bigg)\nonumber\\
&&-\frac{\omega_{B}}{\theta_{B}}\ln\bigg(\int_{\boldsymbol{\gamma}\in
\mathbb
R'_{1}}\big[1+\gamma_{2}P_{B}(\boldsymbol{\epsilon})\big]^{-\frac{\beta_{B}}{2}}\mathit
p_{\boldsymbol{\gamma}}(\gamma_{1},\gamma_{2})d\gamma_{1}d\gamma_{2}\nonumber\\
&&+\int_{\boldsymbol{\gamma}\in \mathbb
R'_{2}}\bigg[1+\frac{\gamma_{2}P_{B}(\boldsymbol{\epsilon})}{1+\gamma_{1}P_{A}(\boldsymbol{\epsilon})}\bigg]^{-\frac{\beta_{B}}{2}}\mathit
p_{\boldsymbol{\gamma}}(\gamma_{1},\gamma_{2})d\gamma_{1}d\gamma_{2}\bigg)\nonumber\\
&&+\bigg(\int_{\boldsymbol{\gamma}\in \mathbb
R'_{1}}\frac{\mu_{A}}{\theta_{A}}\bigg[1+\frac{\gamma_{1}P_{A}(\boldsymbol{\epsilon})}{1+\gamma_{2}P_{B}(\boldsymbol{\epsilon})}\bigg]^{-\frac{\beta_{A}}{2}}\mathit
p_{\boldsymbol{\gamma}}(\gamma_{1},\gamma_{2})d\gamma_{1}d\gamma_{2}\nonumber\\
&&+\int_{\boldsymbol{\gamma}\in \mathbb
R'_{2}}\frac{\mu_{A}}{\theta_{A}}\big[1+\gamma_{1}P_{A}(\boldsymbol{\epsilon})\big]^{-\frac{\beta_{A}}{2}}\mathit
p_{\boldsymbol{\gamma}}(\gamma_{1},\gamma_{2})d\gamma_{1}d\gamma_{2}\bigg) \nonumber\\
&&+\bigg(\int_{\boldsymbol{\gamma}\in \mathbb
R'_{1}}\frac{\mu_{B}}{\theta_{B}}\big[1+\gamma_{2}P_{B}(\boldsymbol{\epsilon})\big]^{-\frac{\beta_{B}}{2}}\mathit
p_{\boldsymbol{\gamma}}(\gamma_{1},\gamma_{2})d\gamma_{1}d\gamma_{2}\nonumber\\
&&+\int_{\boldsymbol{\gamma}\in \mathbb
R'_{2}}\frac{\mu_{B}}{\theta_{B}}\bigg[1+\frac{\gamma_{2}P_{B}(\boldsymbol{\epsilon})}{1+\gamma_{1}P_{A}(\boldsymbol{\epsilon})}\bigg]^{-\frac{\beta_{B}}{2}}\mathit
p_{\boldsymbol{\gamma}}(\gamma_{1},\gamma_{2})d\gamma_{1}d\gamma_{2}\bigg) \nonumber\\
&&+\lambda_{A}\big(\overline{P_{A}}-\mathbb{E}_{\boldsymbol\gamma}[P_{A}(\boldsymbol{\epsilon})]\big)+\lambda_{B}\big(\overline{P_{B}}-\mathbb{E}_{\boldsymbol\gamma}[P_{B}(\boldsymbol{\epsilon})]\big),
\end{eqnarray}

The partial derivative of (\ref{eqn:mac}) with respect to
$P_{A}(\boldsymbol{\epsilon})$ is given by
\begin{eqnarray}\label{eqn:e37}
&&\hspace{-0.3cm}-\frac{\omega_{A}}{\theta_{A}\phi'_{1}}\bigg(\int_{\boldsymbol{\gamma}\in
\mathbb
R'_{1}}-\frac{\beta_{A}\gamma_{1}}{2\big[1+\gamma_{2}P_{B}(\boldsymbol{\epsilon})\big]}\nonumber\\
&&\hspace{-0.3cm}\times\bigg[1+\frac{\gamma_{1}P_{A}(\boldsymbol{\epsilon})}{1+\gamma_{2}P_{B}(\boldsymbol{\epsilon})}\bigg]^{-\frac{\beta_{A}}{2}-1}\mathit
p_{\boldsymbol{\gamma}}(\gamma_{1},\gamma_{2})d\gamma_{1}d\gamma_{2}\nonumber\\
&&\hspace{-0.3cm}+\int_{\boldsymbol{\gamma}\in \mathbb
R'_{2}}-\frac{\beta_{A}\gamma_{1}}{2}\big[1+\gamma_{1}P_{A}(\boldsymbol{\epsilon})\big]^{-\frac{\beta_{A}}{2}-1}\mathit
p_{\boldsymbol{\gamma}}(\gamma_{1},\gamma_{2})d\gamma_{1}d\gamma_{2}\bigg)\nonumber\\
&&\hspace{-0.3cm}-\frac{\omega_{B}}{\theta_{B}\phi'_{2}}\bigg(\int_{\boldsymbol{\gamma}\in
\mathbb
R'_{2}}\frac{\beta_{B}\gamma_{1}\gamma_{2}P_{B}(\boldsymbol{\epsilon})}{2\big[1+\gamma_{1}P_{A}(\boldsymbol{\epsilon})\big]^{2}}\nonumber\\
&&\hspace{-0.3cm}\times\bigg[1+\frac{\gamma_{2}P_{B}(\boldsymbol{\epsilon})}{1+\gamma_{1}P_{A}(\boldsymbol{\epsilon})}\bigg]^{-\frac{\beta_{B}}{2}-1}\mathit
p_{\boldsymbol{\gamma}}(\gamma_{1},\gamma_{2})d\gamma_{1}d\gamma_{2}\bigg)\nonumber\\
&&\hspace{-0.3cm}+\bigg(\int_{\boldsymbol{\gamma}\in \mathbb
R'_{1}}-\frac{\mu_{A}\beta_{A}\gamma_{1}}{2\theta_{A}\big[1+\gamma_{2}P_{B}(\boldsymbol{\epsilon})\big]}\nonumber\\
&&\hspace{-0.3cm}\times\bigg[1+\frac{\gamma_{1}P_{A}(\boldsymbol{\epsilon})}{1+\gamma_{2}P_{B}(\boldsymbol{\epsilon})}\bigg]^{-\frac{\beta_{A}}{2}-1}\mathit
p_{\boldsymbol{\gamma}}(\gamma_{1},\gamma_{2})d\gamma_{1}d\gamma_{2}\nonumber\\
&&\hspace{-0.3cm}+\int_{\boldsymbol{\gamma}\in \mathbb
R'_{2}}-\frac{\mu_{A}\beta_{A}\gamma_{1}}{2\theta_{A}}\big[1+\gamma_{1}P_{A}(\boldsymbol{\epsilon})\big]^{-\frac{\beta_{A}}{2}-1}\mathit
p_{\boldsymbol{\gamma}}(\gamma_{1},\gamma_{2})d\gamma_{1}d\gamma_{2}\bigg)\nonumber\\
&&\hspace{-0.3cm}+\bigg(\int_{\boldsymbol{\gamma}\in \mathbb
R'_{2}}\frac{\mu_{B}\beta_{B}\gamma_{1}\gamma_{2}P_{B}(\boldsymbol{\epsilon})}{2\theta_{B}\big[1+\gamma_{1}P_{A}(\boldsymbol{\epsilon})\big]^{2}}\bigg[1+\frac{\gamma_{2}P_{B}(\boldsymbol{\epsilon})}{1+\gamma_{1}P_{A}(\boldsymbol{\epsilon})}\bigg]^{-\frac{\beta_{B}}{2}-1}\nonumber\\
&&\hspace{-0.3cm}\times\mathit
p_{\boldsymbol{\gamma}}(\gamma_{1},\gamma_{2})d\gamma_{1}d\gamma_{2}\bigg)-\lambda_{A}.
\end{eqnarray}
By differentiating on $P_{B}(\boldsymbol{\epsilon})$ of
(\ref{eqn:mac}), the similar result can be obtained as
\begin{eqnarray}\label{eqn:e38}
&&-\frac{\omega_{A}}{\theta_{A}\phi'_{1}}\bigg(\int_{\boldsymbol{\gamma}\in\mathbb
R'_{1}}\frac{\beta_{A}\gamma_{1}\gamma_{2}P_{A}(\boldsymbol{\epsilon})}{2\big[1+\gamma_{2}P_{B}(\boldsymbol{\epsilon})\big]^{2}}\nonumber\\
&&\times\bigg[1+\frac{\gamma_{1}P_{A}(\boldsymbol{\epsilon})}{1+\gamma_{2}P_{B}(\boldsymbol{\epsilon})}\bigg]^{-\frac{\beta_{A}}{2}-1}\mathit
p_{\boldsymbol{\gamma}}(\gamma_{1},\gamma_{2})d\gamma_{1}d\gamma_{2}\bigg)\nonumber\\
&&-\frac{\omega_{B}}{\theta_{B}\phi'_{2}}\bigg(\int_{\boldsymbol{\gamma}\in
\mathbb
R'_{1}}-\frac{\beta_{B}\gamma_{2}}{2}\big[1+\gamma_{2}P_{B}(\boldsymbol{\epsilon})\big]^{-\frac{\beta_{B}}{2}-1}\nonumber\\
&&\times\mathit
p_{\boldsymbol{\gamma}}(\gamma_{1},\gamma_{2})d\gamma_{1}d\gamma_{2}+\int_{\boldsymbol{\gamma}\in
\mathbb
R'_{2}}-\frac{\beta_{B}\gamma_{2}}{2\big[1+\gamma_{1}P_{A}(\boldsymbol{\epsilon})\big]}\nonumber\\
&&\times\bigg[1+\frac{\gamma_{2}P_{B}(\boldsymbol{\epsilon})}{1+\gamma_{1}P_{A}(\boldsymbol{\epsilon})}\bigg]^{-\frac{\beta_{B}}{2}-1}\mathit
p_{\boldsymbol{\gamma}}(\gamma_{1},\gamma_{2})d\gamma_{1}d\gamma_{2}\bigg)\nonumber\\
&&+\bigg(\int_{\boldsymbol{\gamma}\in \mathbb
R'_{1}}\frac{\mu_{A}\beta_{A}\gamma_{1}\gamma_{2}P_{A}(\boldsymbol{\epsilon})}{2\theta_{A}\big[1+\gamma_{2}P_{B}(\boldsymbol{\epsilon})\big]^{2}}\bigg[1+\frac{\gamma_{1}P_{A}(\boldsymbol{\epsilon})}{1+\gamma_{2}P_{B}(\boldsymbol{\epsilon})}\bigg]^{-\frac{\beta_{A}}{2}-1}\nonumber\\
&&\times\mathit
p_{\boldsymbol{\gamma}}(\gamma_{1},\gamma_{2})d\gamma_{1}d\gamma_{2}\bigg)+\bigg(\int_{\boldsymbol{\gamma}\in
\mathbb
R'_{1}}-\frac{\mu_{B}\beta_{B}\gamma_{2}}{2\theta_{B}}\nonumber\\
&&\times\big[1+\gamma_{2}P_{B}(\boldsymbol{\epsilon})\big]^{-\frac{\beta_{B}}{2}-1}\mathit
p_{\boldsymbol{\gamma}}(\gamma_{1},\gamma_{2})d\gamma_{1}d\gamma_{2}\nonumber\\
&&+\int_{\boldsymbol{\gamma}\in \mathbb
R'_{2}}-\frac{\mu_{B}\beta_{B}\gamma_{2}}{2\theta_{B}\big[1+\gamma_{1}P_{A}(\boldsymbol{\epsilon})\big]}\bigg[1+\frac{\gamma_{2}P_{B}(\boldsymbol{\epsilon})}{1+\gamma_{1}P_{A}(\boldsymbol{\epsilon})}\bigg]^{-\frac{\beta_{B}}{2}-1}\nonumber\\
&&\times\mathit
p_{\boldsymbol{\gamma}}(\gamma_{1},\gamma_{2})d\gamma_{1}d\gamma_{2}\bigg)-\lambda_{B}.
\end{eqnarray}

Let the derivatives equal to zero. If $\boldsymbol\gamma \in \mathbb
R'_{1}$, from (\ref{eqn:e37}), the optimal condition should satisfy
\begin{eqnarray}
&&\int_{\boldsymbol{\gamma}\in\mathbb
R'_{1}}\bigg\{\bigg(\frac{\omega_{A}\beta_{A}}{2\theta_{A}\phi'_{1}}-\frac{\mu_{A}\beta_{A}}{2\theta_{A}}\bigg)\frac{\gamma_{1}}{1+\gamma_{2}P_{B}(\boldsymbol{\epsilon})}\nonumber\\
&&\times\bigg[1+\frac{\gamma_{1}P_{A}(\boldsymbol{\epsilon})}{1+\gamma_{2}P_{B}(\boldsymbol{\epsilon})}\bigg]^{-\frac{\beta_{A}+2}{2}}-\lambda_{A}\bigg\}\mathit
p_{\boldsymbol{\gamma}}(\gamma_{1},\gamma_{2})d\gamma_{1}d\gamma_{2}=0,\nonumber\\
&&
\end{eqnarray}
which gives
\begin{equation}\label{eqn:e39}
\frac{\gamma_{1}(\omega_{A}{\phi'_{1}}^{-1}
-\mu_{A})}{1+\gamma_{2}P_{B}(\boldsymbol{\epsilon})}\bigg[1+\frac{\gamma_{1}P_{A}(\boldsymbol{\epsilon})}{1+\gamma_{2}P_{B}(\boldsymbol{\epsilon})}\bigg]^{-\frac{\beta_{A}+2}{2}}=\delta\lambda_{A}.
\end{equation}
From (\ref{eqn:e38}), we can find the following optimality condition
\begin{eqnarray}
&&\hspace{-0.5cm}\int_{\boldsymbol{\gamma}\in\mathbb
R'_{1}}\bigg\{\bigg(\frac{\omega_{B}\beta_{B}}{2\theta_{B}\phi'_{2}}
-\frac{\mu_{B}\beta_{B}}{2\theta_{B}}\bigg)\gamma_{2}\big[1+\gamma_{2}P_{B}(\boldsymbol{\epsilon})\big]^{-\frac{\beta_{B}+2}{2}}\nonumber\\
&&\hspace{-0.5cm}-\lambda_{B}-\bigg(\frac{\omega_{A}\beta_{A}}{2\theta_{A}\phi'_{1}}
-\frac{\mu_{A}\beta_{A}}{2\theta_{A}}\bigg)\frac{\gamma_{1}\gamma_{2}P_{A}(\boldsymbol{\epsilon})}{\big[1+\gamma_{2}P_{B}(\boldsymbol{\epsilon})\big]^{2}}\nonumber\\
&&\hspace{-0.5cm}\times\bigg[1+\frac{\gamma_{1}P_{A}(\boldsymbol{\epsilon})}{1+\gamma_{2}P_{B}(\boldsymbol{\epsilon})}\bigg]^{-\frac{\beta_{A}+2}{2}}
\bigg\}\mathit
p_{\boldsymbol{\gamma}}(\gamma_{1},\gamma_{2})d\gamma_{1}d\gamma_{2}=0,
\end{eqnarray}
which offers
\begin{eqnarray}\label{eqn:e40}
&&\hspace{-0.28cm}\gamma_{2}(\omega_{B}{\phi'_{2}}^{-1}
-\mu_{B})\big[1+\gamma_{2}P_{B}(\boldsymbol{\epsilon})\big]^{-\frac{\beta_{B}+2}{2}}-\delta\lambda_{B}\nonumber\\
&&\hspace{-0.28cm}-\frac{(\omega_{A}{\phi'_{1}}^{-1}
-\mu_{A})\gamma_{1}\gamma_{2}P_{A}(\boldsymbol{\epsilon})}{\big[1+\gamma_{2}P_{B}(\boldsymbol{\epsilon})\big]^{2}}\bigg[1+\frac{\gamma_{1}P_{A}(\boldsymbol{\epsilon})}{1+\gamma_{2}P_{B}(\boldsymbol{\epsilon})}\bigg]^{-\frac{\beta_{A}+2}{2}}=0.\nonumber\\
\end{eqnarray}
Replacing (\ref{eqn:e39}) into (\ref{eqn:e40}), and after some
calculations, we can obtain the optimal power allocation of
$P_{A}(\boldsymbol{\epsilon})$ and $P_{B}(\boldsymbol{\epsilon})$ in
(\ref{eqn:pa}) and (\ref{eqn:pb}).

\section{Derivation of Optimal Partition Criterion for Two-phase Two-way Relaying}\label{app:partition}

Here we only focus on the derivation for the threshold
$\gamma_{2}^{th}$, while $\gamma_{1}^{th}$ can be obtained by the
same way. Let we write the optimal $\gamma_{2}^{th}$ as a function
of $\gamma_{1}$, i.e., $\gamma_{2}^{th}=f^{*}(\gamma_{1})$, where
$f^{*}(\gamma_{1})$ is the optimal function. We define
$f(\gamma_{1})=f^{*}(\gamma_{1})+sg(\gamma_{1})$, where $s$ is any
constant and $g(\gamma_{1})$ represents arbitrary variation. Thus,
(\ref{eqn:e23}) can be rewritten as
\begin{eqnarray}
&&~\mathcal{J}(f(\gamma_{1}))\nonumber\\
&&\hspace{-0.3cm}=-\frac{\omega_{A}}{\theta_{A}}\ln\bigg(\int_0^\infty\int_0^{f(\gamma_{1})}\bigg[1+\frac{\gamma_{1}P_{A}(\boldsymbol{\epsilon})}{1+\gamma_{2}P_{B}(\boldsymbol{\epsilon})}\bigg]^{-\frac{\beta_{A}}{2}}\mathit
p_{\boldsymbol{\gamma}}(\gamma_{1},\gamma_{2})\nonumber\\
&&\hspace{-0.3cm}\times
d\gamma_{1}d\gamma_{2}+\int_0^\infty\int_{f(\gamma_{1})}^{\infty}\big[1+\gamma_{1}P_{A}(\boldsymbol{\epsilon})\big]^{-\frac{\beta_{A}}{2}}\mathit
p_{\boldsymbol{\gamma}}(\gamma_{1},\gamma_{2})d\gamma_{1}d\gamma_{2}\bigg)\nonumber\\
&&\hspace{-0.3cm}-\frac{\omega_{B}}{\theta_{B}}\ln\bigg(\int_0^\infty\int_0^{f(\gamma_{1})}\big[1+\gamma_{2}P_{B}(\boldsymbol{\epsilon})\big]^{-\frac{\beta_{B}}{2}}\mathit
p_{\boldsymbol{\gamma}}(\gamma_{1},\gamma_{2})d\gamma_{1}d\gamma_{2}\nonumber\\
&&\hspace{-0.3cm}+\int_0^\infty\int_{f(\gamma_{1})}^{\infty}\bigg[1+\frac{\gamma_{2}P_{B}(\boldsymbol{\epsilon})}{1+\gamma_{1}P_{A}(\boldsymbol{\epsilon})}\bigg]^{-\frac{\beta_{B}}{2}}\mathit
p_{\boldsymbol{\gamma}}(\gamma_{1},\gamma_{2})d\gamma_{1}d\gamma_{2}\bigg)\nonumber\\
&&\hspace{-0.3cm}+\bigg(\int_0^\infty\int_0^{f(\gamma_{1})}\frac{\mu_{A}}{\theta_{A}}\bigg[1+\frac{\gamma_{1}P_{A}(\boldsymbol{\epsilon})}{1+\gamma_{2}P_{B}(\boldsymbol{\epsilon})}\bigg]^{-\frac{\beta_{A}}{2}}\mathit
p_{\boldsymbol{\gamma}}(\gamma_{1},\gamma_{2})d\gamma_{1}d\gamma_{2}\nonumber\\
&&\hspace{-0.3cm}+\int_0^\infty\int_{f(\gamma_{1})}^{\infty}\frac{\mu_{A}}{\theta_{A}}\big[1+\gamma_{1}P_{A}(\boldsymbol{\epsilon})\big]^{-\frac{\beta_{A}}{2}}\mathit
p_{\boldsymbol{\gamma}}(\gamma_{1},\gamma_{2})d\gamma_{1}d\gamma_{2}\bigg) \nonumber\\
&&\hspace{-0.3cm}+\bigg(\int_0^\infty\int_0^{f(\gamma_{1})}\frac{\mu_{B}}{\theta_{B}}\big[1+\gamma_{2}P_{B}(\boldsymbol{\epsilon})\big]^{-\frac{\beta_{B}}{2}}\mathit
p_{\boldsymbol{\gamma}}(\gamma_{1},\gamma_{2})d\gamma_{1}d\gamma_{2}\nonumber\\
&&\hspace{-0.3cm}+\int_0^\infty\int_{f(\gamma_{1})}^{\infty}\frac{\mu_{B}}{\theta_{B}}\bigg[1+\frac{\gamma_{2}P_{B}(\boldsymbol{\epsilon})}{1+\gamma_{1}P_{A}(\boldsymbol{\epsilon})}\bigg]^{-\frac{\beta_{B}}{2}}\mathit
p_{\boldsymbol{\gamma}}(\gamma_{1},\gamma_{2})d\gamma_{1}d\gamma_{2}\bigg)\nonumber\\
&&\hspace{-0.3cm}+\lambda_{A}\big(\overline{P_{A}}-\mathbb{E}_{\boldsymbol\gamma}[P_{A}(\boldsymbol{\epsilon})]\big)+\lambda_{B}\big(\overline{P_{B}}-\mathbb{E}_{\boldsymbol\gamma}[P_{B}(\boldsymbol{\epsilon})]\big).
\end{eqnarray}
Intuitively, $\mathcal{J}\big(f(\gamma_{1})\big)$ gains its optimal
value when $s=0$, namely, the derivative of
$\mathcal{J}\big(f(\gamma_{1})\big)$ is equal to zero when $s=0$.
Therefore, to obtain the optimal condition, it is necessary to
satisfy \cite{Arfken}
\begin{equation}
\frac{\emph{d}}{\emph{d}s}\big(\mathcal{J}\big(f(\gamma_{1})\big)\big)\bigg|_{s=0}=0.
\end{equation}
Then, it follows that
\begin{eqnarray}
&&\int_0^\infty\bigg\{-\bigg(\frac{\omega_{A}}{\theta_{A}\phi'_{1}}-\frac{\mu_{A}}{\theta_{A}}\bigg)\bigg(\bigg[1+\frac{\gamma_{1}P_{A}(\boldsymbol{\epsilon})}{1+f^{*}(\gamma_{1})P_{B}(\boldsymbol{\epsilon})}\bigg]^{-\frac{\beta_{A}}{2}}\nonumber\\
&&-\big[1+\gamma_{1}P_{A}(\boldsymbol{\epsilon})\big]^{-\frac{\beta_{A}}{2}}\bigg)-\bigg(\frac{\omega_{B}}{\theta_{B}\phi'_{2}}-\frac{\mu_{B}}{\theta_{B}}\bigg)\nonumber\\
&&\times\bigg(\big[1+f^{*}(\gamma_{1})P_{B}(\boldsymbol{\epsilon})\big]^{-\frac{\beta_{B}}{2}}-\bigg[1+\frac{f^{*}(\gamma_{1})P_{B}(\boldsymbol{\epsilon})}{1+\gamma_{1}P_{A}(\boldsymbol{\epsilon})}\bigg]^{-\frac{\beta_{B}}{2}}\bigg)\bigg\}\nonumber\\
&&\times\mathit
p_{\boldsymbol{\gamma}}(\gamma_{1},f^{*}(\gamma_{1}))g(\gamma_{1})d\gamma_{1}=0.
\end{eqnarray}
Since the above equation needs to be hold for any $g(\gamma_{1})$,
we can obtain
\begin{eqnarray}\label{eqn:e41}
&&\frac{\big[1+\frac{\gamma_{1}P_{A}(\boldsymbol{\epsilon})}{1+f^{*}(\gamma_{1})P_{B}(\boldsymbol{\epsilon})}\big]^{-\frac{\beta_{A}}{2}}-\big[1+\gamma_{1}P_{A}(\boldsymbol{\epsilon})\big]^{-\frac{\beta_{A}}{2}}}
{\big[1+\frac{f^{*}(\gamma_{1})P_{B}(\boldsymbol{\epsilon})}{1+\gamma_{1}P_{A}(\boldsymbol{\epsilon})}\big]^{-\frac{\beta_{B}}{2}}-\big[1+f^{*}(\gamma_{1})P_{B}(\boldsymbol{\epsilon})\big]^{-\frac{\beta_{B}}{2}}}\nonumber\\
&=&\frac{\theta_{A}(\omega_{B}{\phi'_{2}}^{-1}-\mu_{B})}
{\theta_{B}(\omega_{A}{\phi'_{1}}^{-1}-\mu_{A})}\triangleq K.
\end{eqnarray}
Obviously, $K\geq0$ for all cases. Specifically,
$\gamma_{2}^{th}=f^{*}(\gamma_{1})$ should be positive, otherwise it
would not be well-defined in (\ref{eqn:e41}).


\begin{IEEEbiography}
{Cen
Lin} received the B.S. degree in electrical engineering from
Shanghai Jiao Tong University, Shanghai, China, in 2010. He is
currently pursuing a dual M.S degree in electrical and computer
engineering from Shanghai Jiao Tong University and Georgia Institute
of Technology. His research interests include cooperative
communications, physical layer network coding, and resource
allocation for QoS provisioning.
\end{IEEEbiography}

\begin{IEEEbiography}
{Yuan Liu}(S'11)
received the B.S. degree from Hunan University of Science and
Technology, Xiangtan, China, in 2006, and the M.S. degree from
Guangdong University of Technology, Guangzhou, China, in 2009, both
in Communications Engineering and with the highest honors. He is
currently pursuing his Ph.D. degree at the Department of Electrical
Engineering in Shanghai Jiao Tong University. His current research
interests include cooperative communications, network coding,
resource allocation, physical layer security, MIMO and OFDM
techniques.

He is the recipient of the Guangdong Province Excellent Master
Theses Award in 2010. He has been honored as an Exemplary Reviewer
of the \textsc{IEEE Communications Letters}. He is also awarded the
IEEE Student Travel Grant for IEEE ICC 2012. He is a student member
of the IEEE.
\end{IEEEbiography}

\begin{IEEEbiography}
{Meixia Tao}(S'00-M'04-SM'10)
received the B.S. degree in electronic engineering from Fudan
University, Shanghai, China, in 1999, and the Ph.D. degree in
electrical and electronic engineering from Hong Kong University
of Science and Technology in 2003. She is currently an Associate
Professor with the Department of Electronic Engineering, Shanghai
Jiao Tong University, China. From August 2003 to August 2004, she
was a Member of Professional Staff at Hong Kong Applied Science
and Technology Research Institute Co. Ltd. From August 2004 to
December 2007, she was with the Department of Electrical and
Computer Engineering, National University of Singapore, as an
Assistant Professor. Her current research interests include
cooperative transmission, physical layer network coding, resource
allocation of OFDM networks, and MIMO techniques.

Dr. Tao is an Editor for the \textsc{IEEE Wireless Communications
Letters}, an Associate Editor for the \textsc{IEEE Communications
Letters} and an Editor for the \emph{Journal of Communications and
Networks}. She  was on the  Editorial  Board of the \textsc{IEEE
Transactions on Wireless Communications} from 2007 to 2011. She
served as Track/Symposium Co-Chair for APCC09, ChinaCom09, IEEE
ICCCN07, and IEEE ICCCAS07. She has also served as Technical Program
Committee member for various conferences, including IEEE INFOCOM,
IEEE GLOBECOM, IEEE ICC, IEEE WCNC, and IEEE VTC.

Dr. Tao is the recipient of the IEEE ComSoC Asia-Pacific Outstanding
Young Researcher Award in 2009.
\end{IEEEbiography}

\end{document}